\documentclass[11pt,letterpaper]{article}

\usepackage{hyperref}
\usepackage[numbers]{natbib}
\usepackage{xspace}
\usepackage{xcolor}
\usepackage{graphicx}
\usepackage{amsmath,amssymb,amsthm}
\usepackage{algorithm,algpseudocode}
\usepackage[margin=1in]{geometry}
\usepackage{tikz}
\usetikzlibrary{chains,arrows}
\usepackage{dsfont}
\usepackage{todonotes}
\usepackage{mathtools}
\usepackage{color-edits}
\addauthor{rn}{red}     
\addauthor{bk}{blue}    
\addauthor{tk}{green}   

\theoremstyle{plain}

\newtheorem{theorem}{Theorem}
\newtheorem{lemma}{Lemma}

\theoremstyle{definition}
\newtheorem{definition}{Definition}

\newtheorem{remark}{Remark}
\newcommand{\secref}[1]{\S\ref{#1}}
\newcommand{\thmref}[1]{Theorem~\ref{#1}}
\newcommand{\lemref}[1]{Lemma~\ref{#1}}

\newcommand{\appref}[1]{Appendix~\ref{#1}}

\newcommand{\ie}{i.e.,\xspace}

\newcommand{\growingmid}{\mathrel{}\middle|\mathrel{}}

\newcommand{\RR}{\ensuremath{\mathbb{R}}}

\newcommand{\e}{{\mathrm e}}

\newcommand{\numblocks}{q}
\newcommand{\indbound}{p}

\renewcommand{\Pr}[1]{\mbox{\rm\bf Pr}\left[#1\right]}
\newcommand{\Ex}[1]{\mbox{\rm\bf E}\left[#1\right]}

\newcommand{\p}{\pi}
\newcommand{\sig}{\sigma}
\newcommand{\ts}{[n]}

\newcommand{\opt}{\mathcal{O}}
\newcommand{\alg}[3]{\mathsf{ALG}(#1,#2,#3)}
\newcommand{\algset}{\mathcal{A}}
\newcommand{\thresh}{\tau}
\newcommand{\sam}{S}
\newcommand{\nsam}{S^{\mathsf{c}}}

\newcommand{\bind}{\rho}
\newcommand{\val}[1]{v_{\sig^{-1}(#1)}}
\newcommand{\ex}[1]{\mbox{\rm\bf E}\left[#1\right]}

\newcommand{\valfun}{\textrm{val}}

\newcommand{\valoptsam}{\valfun(\opt\cap \sam)}
\newcommand{\valalgsam}{\valfun(\algset\cap \sam)}
\newcommand{\valalgnsam}{\valfun(\algset\cap \nsam)}
\newcommand{\valoptnsam}{\valfun(\opt\cap \nsam)}
\newcommand{\algnumsam}{Q}
\newcommand{\valoptsamt}{\valfun([\opt\cap \sam]_{k/2})}
\newcommand{\eps}{\varepsilon}
\newcommand{\ind}[1]{\mathds{I}\left\{#1\right\}}
\newcommand{\pr}[1]{\Pr{#1}}
\DeclarePairedDelimiter\ceil{\lceil}{\rceil}

\newcommand{\items}{\mathcal{U}}
\newcommand{\algo}{{\text{\sc alg}}}
\newcommand{\natperm}{\pi}
\newcommand{\advperm}{\sigma}
\newcommand{\natpd}{{\underline{\natperm}}}
\newcommand{\advpd}{{\underline{\advperm}}}
\newcommand{\idpd}{{\iota}}
\newcommand{\revperm}{{\rho}}
\newcommand{\dists}[1]{{\mathbf{\Delta}(#1)}}

\newcommand{\npde}{\natpd_e}
\newcommand{\npdacf}{\natpd_{c,1}}
\newcommand{\npdacs}{\natpd_{c,2}}
\newcommand{\algoe}{\algo_e}
\newcommand{\algoacf}{\algo_{c,1}}
\newcommand{\algoacs}{\algo_{c,2}}
\newcommand{\algop}{\algo_p}

\newcommand{\pcs}[1]{{\text{PCS}(#1)}}
\newcommand{\gval}[1][\natpd]{{V}^{#1}}

\newcommand{\hardf}{{g}}
\newcommand{\ptz}{{\kappa}}
\newcommand{\ecc}{{A}}
\newcommand{\proj}{{\text{proj}}}

\newcommand{\trans}{{\intercal}}
\newcommand{\cone}{{\mathcal{C}}}
\newcommand{\conesph}{{\mathcal{S}}}

\begin{document}

\begin{titlepage}
\title{Secretary Problems with Non-Uniform Arrival Order}

\author{%
	Thomas Kesselheim\thanks{Max-Planck-Institut f\"ur Informatik, Campus E1 4, 66123 Saarbr\"ucken, Germany, E-Mail: \texttt{thomas.kesselheim@mpi-inf.mpg.de}.}
	\and 
	Robert Kleinberg\thanks{Department of Computer Science, Cornell University, Gates Hall, Ithaca, NY 14853, USA, E-Mail: \texttt{\{rdk, rad\}@cs.cornell.edu}.}
	\and 
	Rad Niazadeh\footnotemark[2]
}

\date{}
\maketitle

\begin{abstract}
For a number of problems in the theory of online algorithms,
it is known that the assumption that elements arrive in 
uniformly random order enables the design of algorithms
with much better performance guarantees than under worst-case
assumptions. The quintessential example of this phenomenon
is the secretary problem, in which an algorithm attempts to 
stop a sequence at the moment it observes the maximum value
in the sequence. As is well
known, if the sequence is presented in uniformly random order
there is an algorithm that succeeds with probability $1/e$,
whereas no non-trivial performance guarantee is possible if 
the elements arrive in worst-case order.

In many of the applications of online algorithms, it is reasonable
to assume there is some randomness in the input sequence,
but unreasonable to assume that the arrival ordering is uniformly
random.
This work initiates an investigation into relaxations of the 
random-ordering hypothesis in online algorithms, by focusing on
the secretary problem and asking what performance guarantees one
can prove under relaxed assumptions. Toward this end, we 
present two sets of properties of distributions over permutations as sufficient conditions, called the {\em $(p,q,\delta)$-block-independence property} and {\em$(k,\delta)$-uniform-induced-ordering property}. We show these two are asymptotically equivalent by borrowing some techniques from the celebrated {\em approximation theory}. Moreover, we show they both imply
the existence of secretary algorithms with constant probability
of correct selection, approaching the optimal constant $1/e$
as the related parameters of the property tend towards their extreme values. 
Both of these properties are significantly weaker than the usual assumption
of uniform randomness; we substantiate this
by providing several constructions of distributions that satisfy
$(p,q,\delta)$-block-independence. As one application of our
investigation, we prove that $\Theta(\log \log n)$ is 
the minimum entropy of any permutation distribution 
that permits constant probability of correct selection
in the secretary problem with $n$ elements. While our 
block-independence condition is sufficient for constant
probability of correct selection, it is not necessary;
however, we present complexity-theoretic evidence that
no simple necessary and sufficient criterion exists.
Finally, we explore the extent to which the
performance guarantees of other algorithms
are preserved when one relaxes the uniform random ordering
assumption to $(p,q,\delta)$-block-independence,
obtaining a positive result for Kleinberg's multiple-choice
secretary algorithm and a negative result for the weighted
bipartite matching algorithm of Korula and P\'{a}l.
\end{abstract}

\thispagestyle{empty}
\setcounter{page}{0}

\end{titlepage}

\section{Introduction} \label{sec:intro}

A recurring theme in the theory of online algorithms is that
algorithms may perform much better when their input is 
in (uniformly) random order than when the ordering is worst-case. 
The quientessential example of this phenomenon
is the secretary problem, in which an algorithm attempts to 
stop a sequence at the moment it observes the maximum value
in the sequence. As is well
known, if the sequence is presented in uniformly random order
there is an algorithm that succeeds with probability $\frac1e$,
whereas no non-trivial performance guarantee is possible if 
the elements arrive in worst-case order.

In many of the applications of online algorithms, it is reasonable
to assume there is some randomness in the input sequence,
but unreasonable to assume that the input ordering is {\em uniformly}
random. It is therefore of interest to ask which algorithms
have {\em robust} performance guarantees, in the sense that 
the performance guarantee holds not only when the input order
is drawn from the uniform distribution, but whenever the input
order is drawn from a reasonably broad family of distributions
that includes the uniform one.
In other words, we seek relaxations of the standard random-ordering
hypothesis which are weak enough to include many distributions 
of interest, but strong
enough to enable one to prove the same (or qualitatively similar)
performance guarantees for online algorithms.

This work initiates an investigation into relaxations of the 
random-ordering hypothesis in online algorithms, by focusing on
the secretary problem and asking what performance guarantees one
can prove under relaxed assumptions. In the problems we consider
there are three parties: an {\em adversary} that assigns values
to items, {\em nature} which permutes the items into a random
order, and an {\em algorithm} that observes the items and their values
in the order specified by nature. To state our
results, let us say that a distribution over permutations,
is {\em secretary-admissible} (abbreviated {\em s-admissible})
if it is the case that when nature uses this distribution to 
sample the ordering of items,
there exists an algorithm that is guaranteed at least a constant 
probability of selecting the element of maximum value, no matter
what values the adversary assigns to elements. 
If this constant probability approaches $\frac{1}{e}$ as 
the number of elements, $n$, goes to infinity, we say that the
distribution is {\em secretary-optimal (s-optimal)}.
\begin{quotation}
\noindent {\bf Question 1:} {\em What natural properties of a distribution
suffice to guarantee that it is s-admissible? What properties
suffice to guarantee that it is s-optimal?}
\end{quotation}
For example, rather than assuming that ordering of
the entire $n$-tuple of items is uniformly random,
suppose we fix a constant $k$ and assume that for every $k$-tuple of 
distinct items, the relative order in which they
appear in the input sequence is $\delta$-close to uniform. 
Does this imply that the distribution is s-admissible?
In \secref{sec:non-uniform} we formalize this 
{\em $(k,\delta)$-uniform-induced-ordering property (UIOP)},
and we prove that it implies s-admissibility for $k \geq 3$
and approaches s-optimality as $k \to \infty$ and
$\delta \to 0$. To prove this, we relate the uniform-induced-ordering
property to another property, the {\em $(p,q,\delta)$-block-independence
property (BIP)}, which may be of independent interest. Roughly speaking,
the block-independence property asserts that the joint distribution
of arrival times of any $p$ distinct elements, when considered
at coarse enough granularity, is $\delta$-close to $p$ i.i.d.\ samples
from the uniform distribution. While this property may sound much
stronger than the UIOP, we show that it is actually implied by the UIOP
for sufficiently large $k$ and small $\delta$.

To substantiate the notion that these properties are satisfied
by many interesting distributions that are far from uniform, we
show that they apply to several natural families of permutation
distributions, including almost every uniform distribution with
support size $\omega(\log n)$, and the distribution over linear 
orderings defined by taking any $n$ sufficiently ``incoherent''
vectors and projecting them onto a random line.


A distinct but related topic in the theory of computing is
pseudorandomness, which shares a similar emphasis on showing
that performance guarantees of certain classes of algorithms
are preserved when one replaces the uniform distribution over
inputs with suitably chosen non-uniform distributions, specifically
those having low entropy. While our interest in 
s-admissibility and the $(k,\delta)$-UIOP is primarily motivated
by the considerations of robustness articulated earlier, the
analogy with pseudorandomness prompts a natural set of questions.
\begin{quotation}
\noindent {\bf Question 2:} {\em What is the minimum entropy of an 
s-admissible distribution? What is the minimum entropy of 
a distribution that satisfies the $(k,\delta)$-UIOP? Is
there an explicit construction that achieves the minimum
entropy?}
\end{quotation}
In \secref{sec:non-uniform} and \secref{sec:entropy}
we supply matching upper and lower bounds to answer the first
two questions. The answer is the same in both cases, and it is 
surprisingly small: $\Theta(\log \log n)$ bits. Moreover,
$\Theta(\log \log n)$ bits suffice not just for s-admissibility,
but for s-optimality!
We also supply an explicit construction, using Reed-Solomon codes,
of distributions with $\Theta(\log \log n)$ bits of entropy that
satisfy all of these properties.

Given that the $(k,\delta)$-UIOP is a sufficient condition for
s-admissibility, that it is satisfied in every natural construction
of s-admissible distributions that we know of, and that the minimum
entropy of $(k,\delta)$-UIOP distributions matches the minimum
entropy of s-admissible distributions, it is tempting to hypothesize
that the $(k,\delta)$-UIOP (or something very similar) is both
necessary and sufficient for s-admissibility.
\begin{quotation}
\noindent {\bf Question 3:} {\em Find a natural necessary and sufficient condition
that characterizes the property of s-admissibility.}
\end{quotation}
In \secref{sec:complexity} we show that, unfortunately, this is
probably impossible. We construct a strange
distribution over input orderings that is s-admissible,
but any algorithm achieving constant probability of 
correct selection must use a stopping rule that 
cannot be computed by circuits of size $2^{n /\log^2(n)}$.
The construction makes use of a coding-theoretic 
construction that may be of independent interest:
a binary error-correcting code of block length $n$ 
and message length $m = o(n)$, such that if one erases
any $n-2m$ symbols of the received vector, most
messages can still be uniquely decoded
even if $\Omega(m)$ of the remaining $2m$ symbols
are adversarially corrupted.

Finally, we broaden our scope and 
consider other online problems with randomly-ordered inputs.
\begin{quotation}
\noindent {\bf Question 4:} {\em Are the performance guarantees of other
online algorithms in the uniform-random-order model 
(approximately) preserved
when one relaxes the assumption about the input order to
the $(k,\delta)$-UIOP or the $(p,q,\delta)$-BIP? If the 
performance guarantee is not always preserved in general, what 
additional properties of an algorithm 
suffice to ensure that its performance guarantee is preserved?}
\end{quotation}
This is an open-ended question, but we take some initial steps
toward answering it by looking at two generalizations of the
secretary problem: the multiple-choice secretary problem (a.k.a.\ the
uniform matroid secretary problem) and the online bipartite weighted
matching problem. We show that the algorithm of \citet{Kleinberg05}
for the former problem preserves its performance guarantee,
and the algorithm of \citet{KorulaP09} for the latter problem does not.

\paragraph*{ Related Work.}
The secretary problem was solved by 
\citet{Lindley} and \citet{Dynkin63}.
A sequence of papers
relating secretary problems to online
mechanism design \citep{HKP04,Kleinberg05,BIK07}
touched off a flurry of CS research during the
past 10 years. Much of this research has focused
on the so-called {\em matroid secretary problem}, which remains
unsolved despite a string of breakthroughs including
a recent pair of $O(\log \log r)$-competitive 
algorithms \citep{Lachish14,FSZ15}, where $r$ 
is the matroid rank.
Generalizations are known for weighted
matchings in graphs and hypergraphs \citep{DimitrovP12,KRTV13,KorulaP09},
independent sets~\cite{GHKSV14},
knapsack constraints~\cite{BIKK07}, and
submodular payoff functions~\cite{BHZ13,FNS11},
among others. Of particular relevance to our 
work is the {\em free order model} \citep{JSZ13};
our results on the minimum entropy s-admissible
distribution can be regarded as a randomness-efficient
secretary algorithm in the free-order model.

The uniform-random-ordering hypothesis has been applied
to many other problems in online algorithms, perhaps most
visibly to the AdWords problem \citep{DevanurHayes,FHKMS10}
and its generalizations to online linear programming
with packing constraints \citep{AWY14,DJSW11,KRTV14,MR15},
and online convex programming \citep{AD15}.
Applications of the random-order hypothesis in
minimization settings are more rare; see \citep{Meyerson01,MMP01}
for applications in the context of facility location and network
design.

In seeking a middle ground between worst-case and average-case
analysis, our work contributes to a broad-based 
research program going by the name of ``beyond worst-case
analysis'' \citep{bwca-workshop}. In terms of motivation, there
are clear conceptual parallels between our paper and the
work of \citet{MV08}, who study hashing and 
identify hypotheses on the data-generating
process, much weaker than uniform randomness,
under which random hashing using a 2-universal
hash family has provably good performance, although
at a technical level our paper bears no relation to theirs.

The properties of permutation distributions that 
we identify in our work bear a resemblance to
almost $k$-wise independent permutations (e.g., \citep{KNR09}),
but the $(k,\delta)$-UIOP and $(p,q,\delta)$-BIP are 
much weaker, and consequently permutation distributions
satisfying these properties are much more prevalent 
than almost $k$-wise independent permutations.
\paragraph*{Setting and Notations.}
We consider problems in which an algorithm
selects one or more elements from a set $\items$ of $n$ items.
Items are presented sequentially, and
an algorithm may only select items at the 
time when they are presented. In the {\em secretary problem} the items are 
totally ordered by value, and the algorithm is allowed
to select only one element of the input sequence,
with the goal of selecting the item of maximum value.
Algorithms for the secretary problem 
are assumed to be comparison-based\footnote{This 
assumption of comparison-based algorithms
is standard in the literature on secretary problems.
\citet{Samuels} proved that when the input order is 
uniformly random, it is impossible to achieve
probability of correct selection $1/\e + \eps$ for
any constant $\eps>0$, even if the algorithm is allowed
to observe the values.},
meaning their decision whether to 
select the item presented at time $t$
must be based only on the relative ordering
(by value) of the first $t$ elements that arrived. 
Algorithms are evaluated according to their
{\em probability of correct selection}, \ie
the probability of selecting the item of 
maximum value. 

 We assume that the set $\items$ of items
is $[n] = \{1,\ldots,n\}$. The order in which items
are presented is then represented by a permutation
$\natperm$ of $[n]$, where $\natperm(i)$ denotes the
position of item $i$ in the input sequence. Similarly,
the ordering of items by value can be represented by
a permutation $\advperm$ of $[n]$, where $\advperm(j)=i$
means that the $j^{\mathrm{th}}$ largest item is $i$.
Then, the input sequence observed by the algorithm
is completely described by the composition $\natperm \advperm$.


%


\section{Sufficient Properties of Non-Uniform Probability Distributions}
\label{sec:non-uniform}

In \secref{sec:intro}, we introduced two properties of non-uniform
probability distributions which suffice to ensure existence of a
secretary algorithm with constant probability of correct selection.
(In other words, the two properties imply s-admissibility.) 
We begin by formally defining these two properties.

\begin{definition} \label{def:uiop} 
A distribution $\natpd$ over permutations of $[n]$ satisfies
the $(k, \delta)$-uniform-induced-ordering property, abbreviated 
$(k,\delta)$-UIOP, if and
only if, for every $k$ distinct items $x_1, \ldots, x_k \in [n]$, 
if $\pi$ is a random sample from $\natpd$ then
$\Pr{\pi(x_1) < \pi(x_2) < \cdots < \pi(x_k)} \geq (1-\delta) \frac{1}{k!}$. 
\end{definition}


The $(k, \delta)$-uniform-induced-ordering property is a very natural assumption and it is rather easy to show that it is fulfilled by a probability distribution. We will demonstrate this with a few examples in \secref{section:constructions-of-probability-distributions}. However, it is not  clear how to analyze algorithms for secretary problems based on this property. To this end, the more technical $(\indbound,\numblocks,\delta)$-block-independence property is more helpful. We show this by analyzing the classic algorithm for the secretary problem in Section~\ref{section:secretaryanalysis} and the $k$-uniform matroid secretary problem in Section~\ref{section:k-unif}. However, one of our main results in Section~\ref{section:induced-ordering-vs-block-independence} is that these two properties are in fact equivalent, in the limit as the parameters $k,p,q \to \infty$ and $\delta \to 0$.

\begin{definition} \label{def:bip}
Given a positive integer $\numblocks \leq n$, partition $[n]$ into 
$\numblocks$ consecutive disjoint blocks of size between
$\lfloor n/\numblocks \rfloor$ and $\lceil n/\numblocks \rceil$ each, 
denoted by $B_1, \ldots, B_\numblocks \subseteq [n]$.
A permutation distribution $\natpd$ satisfies the 
$(\indbound,\numblocks,\delta)$-block-independence property,
abbreviated $(\indbound,\numblocks,\delta)$-BIP,
if for any distinct $x_1, \ldots, x_\indbound \in \
[n]$, and any $b_1 \ldots, b_\indbound \in [\numblocks]$
\[
\Pr{\bigwedge_{j \in [\indbound]} \pi(x_i) \in B_{b_i}} \geq (1 - \delta) \left(\frac{1}{\numblocks}\right)^\indbound \enspace,
\]
Note that $b_1 \ldots, b_\indbound$ do not necessarily have to be distinct. To simplify notation, given a permutation $\pi$ of $[n]$, 
we define a function $\pi^B\colon \items \to [\numblocks]$ by setting $\pi^B(x) = i$ if and only if $\pi(x) \in B_i$ for all $x \in \items$.
\end{definition}

\subsection{Secretary Algorithms and the $(\indbound,\numblocks,\delta)$-block-independence property}
\label{section:secretaryanalysis}
Next, we will analyze the standard threshold algorithm for the secretary problem under probability distributions that only fulfill the $(\indbound,\numblocks,\delta)$-block-independence property rather than being uniform. The algorithm only observes the first $\frac{n}{\e}$ items. Afterwards, it accepts the first item whose value exceeds all values seen up to this point. Under a uniform distribution, this algorithm picks the best items with probability at least $\frac{1}{\e} - o(1)$. We show that already for small constant values of $\indbound$ and $\numblocks$ and rather large constant values of $\delta$ this algorithm has constant success probability. At the same time, for large $\indbound$ and $\numblocks$ and small $\delta$, the probability converges to $\frac{1}{\e}$.

\begin{theorem}
\label{theorem:secretaryanalysis}
Under a $(\indbound,\numblocks,\delta)$-block-independent probability distribution, the standard secretary algorithm picks the best item with probability at least $\frac{1}{\e} - \frac{\e + 1}{\numblocks} - \delta - \left( 1 - \frac{1}{\e} \right)^{\indbound-1}$.
\end{theorem}

\begin{proof}[Proof Sketch]
Let $T = \lfloor \frac{\numblocks}{\e} \rfloor$ denote the index of the block in which the threshold is located. Furthermore, let $x_j \in \items$ be the $j$th best item. We condition on the event that $x_1$ comes in block with index $i$. To ensure that our algorithm picks this item, it suffices that $x_2$ comes in blocks $1, \ldots, T-1$. Alternatively, we also pick $x_1$ if the $x_2$ comes in blocks $i+1, \ldots, \numblocks$ and $x_3$ comes in blocks $1, \ldots, T-1$. Continuing this argument, we get
\[
\Pr{\text{correct selection}} \geq \sum_{i=T+1}^\numblocks \sum_{j=2}^\indbound \Pr{ \pi^B(x_1) = i, \pi^B(x_2), \ldots, \pi^B(x_{j - 1}) > i, \pi^B(x_j) < T } \enspace.
\]
Note that the $(\indbound,\numblocks,\delta)$-BIP implies the 
$(\indbound',\numblocks,\delta)$-BIP for any $\indbound' < \indbound$,
simply by marginalizing over the remaining indices in the tuple. This gives us:
\[
\Pr{\text{correct selection}} \geq \sum_{i=T+1}^\numblocks \sum_{j=2}^\indbound (1 - \delta) \frac{1}{\numblocks} \left( \frac{\numblocks - i}{\numblocks} \right)^{j-2} \frac{T - 1}{\numblocks} \enspace,
\]
and the lemma follows after manipulating the expression on the right side
and applying some standard bounds.
\end{proof}

\subsection{Relationship Between the Two Properties}
\label{section:induced-ordering-vs-block-independence}
We will show that the two properties 
defined in the preceding section are in some sense equivalent
in the limit as the parameters $k,p,q \to \infty$ and $\delta \to 0$.
(For $k=2$, a distribution satisfying $(k,\delta)$-UIOP is not 
even necessarily s-admissible---this is an easy consequence of the
lower bound in \secref{sec:entropy} and the fact that the $(2,0)$-UIOP
is achieved by a distribution with support size 2, that uniformly 
randomizes between a single permutation and its reverse. Already
for $k=3$ and any constant $\delta<1$, the $(k,\delta)$-UIOP implies
s-admissibility; this is shown in Appendix~\ref{app:3delta}.)

Our first result is relatively straightforward: Any probability distribution that fulfills the $(\indbound,\numblocks,\delta)$-BIP also fulfills the $(\indbound,\delta + \frac{\numblocks^2}{\indbound})$-UIOP. The (easy) proof is deferred to Appendix~\ref{app:rbtp}.
\begin{theorem} 
\label{theorem-BIP}
If a distribution over permutation fulfills the $(\indbound,\numblocks,\delta)$-BIP, then it also fulfills the $(\indbound,\delta + \frac{\indbound^2}{\numblocks})$-UIOP.
\end{theorem}
The other direction is far less obvious. Observe that the $(k,\delta)$-uniform-induced-ordering property works in a purely local sense: even for a single item $x \in \items$, the distribution of its position $\pi(x)$ can be far from uniform. For example, the case $k=2$ is even fulfilled by a two-point distribution that only include one permutation and its reverse. Then $\pi(x)$ can only attain two different values. Nevertheless, we have the following result.
\begin{theorem} \label{thm:uiop-implies-bip}
If a distribution over permutation fulfills the $(k,\delta)$-uniform-induced-ordering property, then it also satisfies  $(p,q,\delta)$-block-independence property for $p=o(k^{\frac{1}{5}})$, $q=O({k^{\frac{1}{5}}})$ as $k$ goes to infinity.
\end{theorem}
The proof applies the \emph{theory of approximation of functions}, which addresses the question of how well one can approximate arbitrary functions by polynomials. The main insight underlying the proof is the following. If $\natpd$ satisfies the $(k,\delta)$-UIOP, then for any $k$-tuple of distinct elements $x_1,\ldots,x_k$ if one defines random variables $X_i \triangleq \natperm(x_i)/n$, then the expected value of any monomial of total degree $k/2$ in the variables $\{X_i\}$ approximates the expected value of that same monomial under the distribution of a uniformly-random permutation. With this lemma in hand, proving \thmref{thm:uiop-implies-bip} becomes a matter of quantifying how well the indicator function of a (multi-dimensional) rectangle can be approximated by low-degree polynomials. Approximation theory furnishes such estimates readily. To make the proof sketch concrete, we start by some definitions and notations from approximation theory; see, e.g., the textbook by~\citet{Carothers}. 
\begin{definition}[\cite{Carothers}]
If $f$ is any bounded function over $[0,1]$, we define the sequence of \emph{Bernstein polynomials} for $f$ by
\begin{equation}
(B_d(f))(x)=\sum_{k=0}^df(k/d) {d\choose k}x^k(1-x)^{d-k},~~~~0\leq x\leq 1.
\end{equation} 
\end{definition}
\begin{remark}
$B_d(f)$ is a polynomial of degree at most $d$.
\end{remark}
\begin{definition}[\cite{Carothers}]
The \emph{modulus of continuity} of a bounded function $f$ over $[a,b]$ is defined by 
\begin{equation}
\omega_f(\delta)=\sup\{\lvert f(x_1)-f(x_2)\rvert: x_1,x_2\in [a,b], \lvert x_1-x_2\rvert\leq \delta\}
\end{equation}
\end{definition}
\begin{remark}
Bounded function $f$ is \emph{continuous} over interval $[a,b]$ if and only if $\omega_f(\delta)=O(\delta)$. Moreover, $f$ is \emph{uniformly continuous} if and only if $\omega_f(\delta)=o(\delta)$. 
\end{remark}
We are now ready to state our main ingredient, i.e. \emph{Bernstein's approximation theorem},  which shows bounded functions with enough continuity are well approximated by Bernstein polynomials.
\begin{theorem}[\cite{Carothers}]
\label{lemma:approximationtheory}
For any bounded function $f$ over $[0,1]$ we have
\begin{equation}
\lVert f- B_d(f) \rVert_{\infty} \leq \frac{3}{2}\omega_f
\left(\frac{1}{\sqrt{d}} \right)
\end{equation}
where for any bounded functions $f_1$ and $f_2$, $\lVert f_1-f_2\rVert_{\infty}\triangleq \sup\{\lvert f_1(x)-f_2(x)\rvert : x\in[0,1]\}$.
\end{theorem}

\begin{proof}[Proof of Theorem~\ref{thm:uiop-implies-bip}]
To prove our claim, we start by showing $(k,\delta)$-uniform-induced-ordering property forces the arrival time of items to have almost the same higher-order moments as uniform independent random variables. More precisely, we have the following lemma (the proof is provided in Appendix~\ref{app:rbtp}).
\begin{lemma}
\label{approxtheory:lemma1}
Suppose $\pi$ is drawn from a permutation distribution satisfying the $(k,\delta)$-uniform-induced-ordering property, and $\{x_1,\ldots,x_p\}$ is an arbitrary set of $p$ items. Let $\phi:[n]\rightarrow\{i/n:i\in [n]\}$ be a uniform random mapping, and random variables $X_i\triangleq \pi (x_i)/n$ for all $i\in[p]$. Then for every $k_i\leq \frac{k}{2p}$ we have $\Ex{\prod_{i=1}^p X_i^{k_i}} \geq(1-\delta)\Ex{\prod_{i=1}^p {\phi(i)}^{k_i}}$.
\end{lemma}
Given Lemma~\ref{approxtheory:lemma1}, roughly speaking the key idea for the rest of the proof is looking at probabilities as the expectation of the indicator functions, and then trying to approximate the indicator functions by polynomials. Now, to compute probabilities all we need are moments, which due to  Lemma~\ref{approxtheory:lemma1} are almost equal to those of uniform independent random variables. Rigorously, we prove the following probabilistic lemma using this idea. (The proof is provided in Appendix~\ref{app:rbtp}).
\begin{lemma}
\label{approxtheory:lemma2}
Let $\phi:[n]\rightarrow\{i/n:i\in [n]\}$ be a uniform random mapping. Furthermore, let $X_1,X_2,\ldots,X_p$ be random variables over $[0,1]$ such that for every $k_i\leq d$ we have $\Ex{\prod_{i=1}^p X_i^{k_i}} \geq\Ex{\prod_{i=1}^p {\phi(i)}^{k_i}}(1-\delta)$. Then for any disjoint intervals  $\{(a_i,b_i)\}_{i=1}^p$ of $[0,1]$ where $a_i$ and $b_i$ are multiples of $1/n$ and $\lvert b_i-a_i \rvert\geq d^{-\frac{1}{4}}$, we have: 
\begin{equation}
\Pr{\bigwedge_{i=1}^{p}\left(X_i\in [a_i,b_i]\right)}\geq\left(\prod_{i=1}^p{(b_i-a_i)}\right)(1-\delta)-\frac{7p}{d^{1\!/\!4}}
\end{equation}
\end{lemma}
Now, by combining Lemma~\ref{approxtheory:lemma1} and Lemma~\ref{approxtheory:lemma2}, 
we check the $(p,q,\delta)$-block-independence property. Start by setting $d=\frac{k}{2p}$. By Lemma~\ref{approxtheory:lemma2} the probability approximation error from what desired  will be $O\left(\frac{p}{d^{1\!/\!4}}\right)$=$O\left(\frac{p^{5\!/\!4}}{k^{1\!/\!4}}\right)$. This error goes to zero as $k\rightarrow \infty$ if we set $p=o(k^{\frac{1}{5}})$. Moreover, we need $\lvert b_i-a_i \rvert\geq d^{-\frac{1}{4}}$. So, $\frac{1}{q}=\Omega\left(\frac{1}{d^{1\!/\!4}}\right)$. As $d=\omega(k^{\frac{4}{5}})$, if we set $q=O({k^{\frac{1}{5}}})$ we are fine. This completes the proof.
\end{proof}
\subsection{Constructions of Probability Distributions Implying the Properties}
\label{section:constructions-of-probability-distributions}

\subsubsection{Randomized One-Dimensional Projections}
\label{sec:1d-proj}

In this section we present one natural construction
leading to a distribution that satisfies the $(k,\delta)$-UIOP.
The starting point for the construction is an $n$-tuple of
vectors $x_1,\ldots,x_n \in \RR^d$. If one sorts these vectors
according to a random one-dimensional projection (i.e., ranks
the vectors in increasing order of $w \cdot x_i$, for a random
$w$ drawn from a spherically symmetric distribution), 
when does the resulting random ordering satisfy the
$(k,\delta)$-UIOP? Note that if any $k$ of these vectors comprise 
an orthonormal $k$-tuple and one ranks them in increasing
order of $w \cdot x_i$, where $w$ is drawn from a spherically
symmetric distribution, then a trivial symmetry argument
shows that the induced ordering of the $k$ vectors is uniformly
random. Intuitively, then, if the vectors $x_1,\ldots,x_n$
are sufficiently ``incoherent'', then any $k$-tuple of them
should be nearly orthonormal and their induced ordering when
projected onto the 1-dimensional subspace spanned by $w$ should
be approximately uniformly random. The present section is devoted
to making this intuition quantitative. We begin by 
recalling the definition of the restricted
isometry property~\citep{CandesTao}.

\begin{definition} \label{def:rip}
A matrix $X$ satisfies the {restricted isometry property (RIP)}
of order $k$ with restricted isometry constant $\delta_k$ if 
the inequalities
\[
(1 - \delta_k) \|x\|^2 \leq
\| X_T x \|^2 \leq 
(1 + \delta_k) \|x\|^2
\]
hold for every submatrix $X_T$ composed of $|T| \leq k$ columns of $X$ and 
every vector $x \in \RR^{|T|}$. Here $\| \cdot \|$ denotes the Euclidean
norm.
\end{definition}

Several random matrix distributions are known to give rise
to matrices satisfying the RIP with high probability. The
simplest such distribution
is a random $d$-by-$n$ matrix with i.i.d.\ entries
drawn from the normal distribution $\mathcal{N} \big( 0,\frac{1}{d} \big)$.
It is known~\citep{baraniuk2008simple,CandesTao} that, with high
probability, such a matrix
satisfies the RIP of order $k$ with restricted isometry constant 
$\delta$ provided that $d = \Omega \big( \frac{k \log n}{\delta^2} \big)$.
Even if the columns $x_1,\ldots,x_n$ of $X$ are not random, if they
are sufficiently ``incoherent'' unit vectors, meaning that 
$x_i \cdot x_j = 1$ if $i=j$ and $x_i \cdot x_j < \delta_k/k$
otherwise, then $X$ satisfies the RIP. Using this idea,  we prove the following theorem (with proof provided in Appendix~\ref{app:rip}).

\begin{theorem} \label{thm:rip}
Let $x_1,\ldots,x_n$ be the columns of a matrix that satisfies 
the RIP of order $k$ with restricted isometry constant 
$\delta_k = \frac{\delta}{k}$. If $w$ is drawn at random
from a spherically symmetric distribution and we use $w$
to define a permutation of $[n]$ by sorting its elements in
order of increasing $w \cdot x_i$, the resulting distribution
over $S_n$ satisfies the $(k,\delta)$-UIOP.
\end{theorem}

\subsubsection{Constructions with Low Entropy}
This subsection presents two constructions showing
that there exist permutation distributions with
entropy $\Theta(\log \log n)$ satisfying
the $(k,\delta)$-UIOP for arbitrarily large constant $k$
and arbitrarily small constant $\delta$. 
The proof of the first result is an easy application of 
the probabilistic method (which is in Appendix~\ref{app:small-support}). The proof of the second result
uses Reed-Solomon codes to supply an explicit construction.

\begin{theorem} \label{thm:rlec}
Fix some $\xi \geq \frac{2 (k+1)!}{\delta^2} \ln n$.
If $S$ is a random $\xi$-element 
multiset of permutations $\pi\colon [n] \to [n]$,
then the uniform distribution over $S$ fulfills 
the $(k,\delta)$-UIOP with probability at least $1-\frac1n$.
\end{theorem}

\begin{theorem} \label{thm:coding-construction}
There is a distribution over permutations that has entropy $O(\log \log n)$ and fulfills the $(k, \delta)$-uniform-induced-ordering property where $\delta = O(\frac{k^2}{\log \log \log n})$.
\end{theorem}

To derive Theorem~\ref{thm:coding-construction}, we start by proving the following lemma.

\begin{lemma}
For large enough $n\in \mathbb{N}$ and some $\ell = \Omega(\log^2 \log \log n)$, there is a distribution over functions $f\colon \items \to [\ell]$ with entropy $O(\log \log n)$ such that for any $x, x' \in \items$, $x \neq x'$, we have $\Pr{f(x) = f(x')} = O(\frac{1}{\log \log \log n})$.
\end{lemma}

\begin{proof}
We will define a function $f$, parameterized by $\alpha_1$, $\alpha_2$, and $\alpha_3$, as a composition of 8 functions, which are mostly injective.

For $i = 1, 2, 3$, let $K_i = \log^{(i)} n$ and $q_i$ be a prime power such that  $2K_i^2+1\geq q_i\geq K_i^2 + 1$ (Note that for large enough $n$, we can always find a prime power between $K_i^2+1$ and $2K_i^2+1$). Let $\alpha_i$ be drawn independently uniformly from $[q_i - 1]$. This is the only randomization involved in the construction. It has entropy $\log(q_1 - 1) + \log(q_2 - 1) + \log(q_3 - 1)$.

Let $C_i$ be a Reed-Solomon code of message length $K_i$ and alphabet size $q_i$. This yields block length $N_i = q_i - 1$ and distance $d_i = N_i - K_i + 1= q_i - K_i $. In other words, $C_i$ is a function $C_i\colon D_i \to R_i$ with $D_i = [q_i]^{K_i}$ and $R_i = [q_i]^{N_i}$ such that for any $w, w' \in D_i$ with $w \neq w'$, we $C_i(w)$ and $C_i(w')$ differ in at least $d_i$ components.

Furthermore $\alpha_i$ defines one position in each code-word $R_i$. Given $\alpha_i$, let $h_i\colon R_i \to [q_i]$, be the projection of a code-word $w$ of $C_i$ to its $\alpha_i$th component, i.e., $h_i(w) = w_{\alpha_i}$.

Finally, we observe that $\lvert D_{i + 1} \rvert = q_{i + 1}^{K_{i + 1}} \geq 5^{K_{i + 1}} \geq 2(2^{K_{i + 1}})^2 + 1 \geq q_i$. So there is an injective mapping $g_i\colon [q_i] \to D_{i + 1}$, mapping alphabet symbols of $C_i$ to messages of $C_{i + 1}$.

Overall, this defines a function $f = h_3 \circ C_3 \circ g_2 \circ h_2 \circ C_2 \circ g_1 \circ h_1 \circ C_1$, mapping values of $D_1$ to $[q_3]$.

Let $f_i = g_i \circ h_i \circ C_i \circ f_{i - 1}$.

Now let $w, w' \in D_1$, $w \neq w'$. Observe that all functions except for the $h_i$ are injective. Therefore the event $f(w) = f(w')$ can only occur if $h_i(C_i(f_{i-1}(w))) = h_i(C_i(f_{i-1}(w')))$ for some $i$. As $C_i$ is a Reed-Solomon code with distance $d_i$, $C_i(f_{i-1}(w))$ and $C_i(f_{i-1}(w'))$ differ in at least $d_i$ components. Therefore, the probability that $h_i(C_i(f_{i-1}(w))) \neq h_i(C_i(f_{i-1}(w')))$ is at least $\frac{d_i}{N_i}$.

By union bound, the combined probability that this does not hold for one $i$ is bounded by
\[
\Pr{\bigwedge_{i=1}^3 h_i(C_i(f_{i-1}(w))) = h_i(C_i(f_{i-1}(w'))) } \leq \sum_{i=1}^3 \left( 1 - \frac{d_i}{N_i} \right) \leq 3 \left( 1 - \frac{d_3}{N_3} \right) \leq \frac{3}{K_3} \enspace.
\]
\end{proof}

\begin{proof}[Proof of Theorem~\ref{thm:coding-construction}]
By the above lemma, there are constants $c_1$, $c_2$, $c_3$ such that the following condition is fulfilled. For some $\ell = c_1 \log^2 \log n$, there is a distribution over functions $f\colon \items \to [\ell]$ with entropy $c_2 \log \log n$ such that for any $x, x' \in \items$, $x \neq x'$, we have $\Pr{f(x) = f(x')} \leq \frac{c_3}{\log \log \log n}$.

Draw a permutation $\pi'\colon [\ell] \to [\ell]$ uniformly at random and define the permutation $\pi\colon \items \to [n]$ by using $\pi' \circ f$ and extending it to a full permutation arbitrarily.

Let $x_1, \ldots, x_k$ be distinct items from $\items$. Conditioned on $f(x_i) \neq f(x_j)$ for all $i \neq j$, we have $\pi(x_1) < \pi(x_2) < \ldots < \pi(x_k)$ with probability $\frac{1}{k!}$. Furthermore, applying a union bound in combination with the above lemma, the probability that there is some pair $i \neq j$, with $f(x_i) = f(x_j)$ is at most $k^2 \frac{c_2}{\log \log \log n}$. Therefore, the overall probability that $\pi(x_1) < \pi(x_2) < \ldots < \pi(x_k)$ is at least $(1 - \frac{c_2 k^2}{\log \log \log n}) \frac{1}{k!}$.

The entropy of the distribution that determines $\pi$ is $c_2 \log \log n + \log(\ell!) = O(\log \log n)$.
\end{proof}

\section{Tight Bound on Entropy of Distribution} \label{sec:entropy}
One of the consequences of the previous section is the fact that there are s-admissible---in fact, even s-optimal---distributions with entropy $O(\log \log n)$. In this section, we show that this bound is actually tight. We show that \emph{every} probability distribution of entropy $o(\log \log n)$  is not s-admissible. The crux of the proof lies in defining a notion of ``semitone sequences''---sequences which satisfy a property similar to, but weaker than, monotonicity---and showing that an adversary can exploit the existence of long semitone sequences to force every algorithm to have a low probability of success.

\begin{theorem}
\label{theorem:lowerboundentropy}
A permutation distribution $\natpd$ of entropy $H = o(\log \log n)$ cannot be s-admissible.
\end{theorem}

 Here is the proof sketch.  We use the fact for distributions of entropy $H$ there is a subset of the support of size $k$ that is selected with probability at least $1 - \frac{8 H}{\log(k - 3)}$. It then suffices to show that if the distribution's support size is at most $k$, then any algorithm's probability of success against a worst-case adversary is at most $\frac{k + 1}{\log n}$. The theorem then follows by setting $k = \sqrt{\log n}$. To bound the algorithm's probability of success, we introduce the notion of \emph{semitone sequences}, defined recursively as follows: an empty sequence is semitone with respect to any permutation $\pi$, and  a sequence $(x_1,\dots,x_s)$ is {\em semitone} w.r.t.\ $\pi$ if $\pi(x_s)\in\{\min_{i\in [s]} \pi(x_i), \max_{i\in [s]} \pi (x_i)\} $ and $(x_1,\dots, x_{s-1})$ is semitone w.r.t.\ $\pi$. We will show that given $k$ arbitrary permutations of $[n]$, there is always a sequence of length $\frac{\log n}{k + 1}$ that is semitone with respect to all $k$ permutations. Later on, we show how an adversary can exploit this sequence to make any algorithm's success probability small. To make the above arguments concrete, we start by this lemma.

\begin{lemma}
\label{low1:lemma-short}
Suppose $\Pi=\{\pi_1,\dots,\pi_k\}$, where each $\pi_i$ is a permutation over $\items$. Then there exists a sequence $(x_1,\ldots,x_s)$ that is semitone with respect to each $\pi_i$ and $s>\frac{\log n}{k + 1}$.
\end{lemma}

\begin{proof}
For a fixed permutation $\pi_i$ and a fixed item $y \in \items$, we define a function $h^y_i\colon \items \setminus \{ y \} \to \{ 0, 1 \}$ that indicates whether $\pi_i$ maps $x$ is to a higher than $y$ or not. Formally,
\[
h^y_i(x) = \left\{
	\begin{array}{ll}
		0  & \text{if } \pi_i(x)<\pi_i(y)  \\
		1 & \text{if } \pi_i(x) > \pi_i(y)
	\end{array} \right. \enspace.
\]
Still keeping one item $y \in \items$ fixed, we now get a $k$-dimensional vector by concatenating the values for different $\pi_i$. This way, we obtain a hash function $\mathbf{h}^y\colon \mathcal{U}\backslash\{y\}\rightarrow\{0,1\}^k$, where $\mathbf{h}^y(x)=\left( h^{y}_1(x),\ldots,h^{y}_k(x)\right)$.

Starting from $U^{(0)} = \items$, we now construct a sequence of nested subsets $U^{(0)} \supseteq U^{(1)} \supseteq \ldots$ iteratively. At iteration $t+1$, given set $U^{(t)} \neq \emptyset$, we do the following. For an arbitrary element $x_{s-t}$ of $U^{(t)}$, we hash each element of $U^{(t)}\backslash \{x_{s-t}\}$ to a value in $\{0,1\}^k$ by using $\mathbf{h}^{x_{s-t}}$. Now $U^{(t+1)}\subseteq U^{(t)}\backslash \{x_{s-t}\}$ is defined to be the set of occupants of the most occupied hash bucket in $\{0,1\}^k$.

Note that if we place $x_{s-t}$ at the end of any semitone sequence in $U^{(t+1)}$ it will remain semitone with respect to each $\pi_i$. This in turn implies that for any $t'$ the sequence $(x_{1}, \ldots, x_{t'})$ is semitone with respect to all $\pi_i$. 

It now remains to bound the length of the sequence $(x_1, \ldots, x_s)$ we are able to generate. We achieve length $s$ if and only if $U^{(s)}$ is the first empty set. At iteration $t$ of the above construction, we have $\lvert U^{(t)}\rvert-1$ elements to hash  and we have $2^k$ hash buckets, so $\lvert U^{(t+1)}\rvert\geq (\lvert U^{(t)}\rvert-1)2^{-k}\geq \lvert U^{(t)}\rvert 2^{-(k+1)}$ and therefore $\lvert U^{(t)}\rvert \geq 2^{-t(k+1)} \lvert U^{(0)}\rvert = 2^{-t(k+1)} n$. As $\lvert U^{(s)} \rvert < 1$, this implies $2^{-s(k+1)} n < 1$. So $s > \frac{\log n}{k + 1}$.
\end{proof}
We now turn to showing that
an adversary can exploit a semitone sequence and force any algorithm to only have $\frac1s$ probability of success. To show this we look at the performance of the best deterministic algorithm against a particular distribution over assignment of values to items.
\begin{lemma}
\label{low3:lemma-short}
Let $\mathcal{V}=\{1,2,\ldots,s\}$. Assign values from $\mathcal{V}$ to items $(x_1,\dots,x_s)$ at random by 
\[
\textrm{value}(x_s)=\begin{cases}
		\max (\mathcal{V} ) & \text{with probability }  1/s \\
		\min (\mathcal{V}) & \text{with probability } 1-1/s
	\end{cases}
\]
and then assigning values from $\mathcal{V}\backslash \{\textrm{value}(x_s)\}$ to items $(x_1,\dots,x_{s-1})$ recursively. Assign a value $0$ to all other items. 

Consider an arbitrary algorithm following permutation $\pi$ such that $(x_1,\dots,x_s)$ is semitone with respect to $\pi$. This algorithm selects the best item with probability at most $\frac{1}{s}$.
\end{lemma}
\begin{proof}
Fixing some (deterministic) algorithm and permutation $\pi$, let $\mathcal{A}_t$ be the event that the algorithm selects any item among $x_1, \ldots, x_t$ and let $\mathcal{B}_t$ be  the event that the algorithm selects the best item among $x_1, \ldots, x_t$. We will show by induction that $\Pr{\mathcal{B}_t} = \frac{1}{t} \Pr{\mathcal{A}_t}$. This will imply $\Pr{\mathcal{B}_s} = \frac{1}{s} \Pr{\mathcal{A}_s} \leq \frac{1}{s}$.

For $t=1$ this statement trivially holds. Therefore, let us consider some $t > 1$. By induction hypothesis, we have $\Pr{\mathcal{B}_{t - 1}} = \frac{1}{t - 1} \Pr{\mathcal{A}_{t - 1}}$. As $(x_1, \ldots, x_t)$ is semitone with respect to $\pi$, $x_t$ either comes before or after all $x_1, \ldots, x_{t-1}$. We distinguish these two cases.

\emph{Case 1: $x_t$ comes before all $x_1, \ldots, x_{t-1}$.} The algorithm can decide to accept $x_t$ (without seeing the items $x_1, \ldots, x_{t-1}$). In this case, we have $\mathcal{A}_t$ for sure. We only have $\mathcal{B}_t$ if $x_t$ gets a higher value than $x_1, \ldots, x_{t-1}$. By definition this happens with probability $\frac{1}{t}$. So, we have $\Pr{\mathcal{B}_t} = \frac{1}{t} \Pr{\mathcal{A}_t}$. The algorithm can also decide to reject $x_t$. Then $\mathcal{A}_t$ if and only if $\mathcal{A}_{t - 1}$. Furthermore, $\mathcal{B}_t$ if and only if $\mathcal{B}_{t - 1}$ and $x_t$ does not get the highest value among $x_1, \ldots, x_t$. These events are independent, so $\Pr{\mathcal{B}_t} = (1 - \frac{1}{t}) \Pr{\mathcal{B}_{t - 1}}$. Applying the induction hypothesis, we get $\Pr{\mathcal{B}_t} = (1 - \frac{1}{t}) \Pr{\mathcal{B}_{t - 1}} = \frac{t-1}{t} \frac{1}{t-1}\Pr{\mathcal{A}_{t - 1}} = \frac{1}{t} \Pr{\mathcal{A}_t}$.

\emph{Case 2: $x_t$ comes after all $x_1, \ldots, x_{t-1}$.} When the algorithm comes to $x_t$, it may or may not have selected an item so far. If it has already selected an item ($\mathcal{A}_{t-1}$), then this element is the best among $x_1, \ldots, x_t$ with probability $\Pr{\mathcal{B}_{t-1} \growingmid \mathcal{A}_{t-1}} = \frac{1}{t-1}$ by induction hypothesis. Independent of these events, $x_t$ is worse than the best items among $x_1, \ldots, x_{t-1}$ with probability $1 - \frac{1}{t}$. Therefore, we get $\Pr{\mathcal{B}_t \growingmid \mathcal{A}_{t - 1}} = \frac{1}{t-1} \frac{t-1}{t} = \frac{1}{t}$. It remains the case that the algorithm selects item $x_t$ ($\mathcal{A}_t \setminus \mathcal{A}_{t-1}$). This item is the better than $x_1, \ldots, x_{t-1}$ with probability $\frac{1}{t}$. That is, $\Pr{\mathcal{B}_t \growingmid \mathcal{A}_t \setminus \mathcal{A}_{t - 1}} = \frac{1}{t}$. In combination, we have $\Pr{\mathcal{B}_t} = \Pr{\mathcal{A}_{t - 1}} \Pr{\mathcal{B}_t \growingmid \mathcal{A}_{t - 1}} + \Pr{\mathcal{A}_t \setminus \mathcal{A}_{t - 1}} \Pr{\mathcal{B}_t \growingmid \mathcal{A}_t \setminus \mathcal{A}_{t - 1}} = \Pr{\mathcal{A}_{t - 1}} \frac{1}{t} + \Pr{\mathcal{A}_t \setminus \mathcal{A}_{t - 1}} \frac{1}{t} = \frac{1}{t} \Pr{\mathcal{A}_t}$.
\end{proof}

Now, to show Theorem~\ref{theorem:lowerboundentropy}, we first give a bound in terms of the support size of the distribution. In fact.  Lemmas~\ref{low1:lemma-short} and \ref{low3:lemma-short} with Yao's principle then imply that any algorithm's probability of success against a worst-case adversary is at most $\frac{k + 1}{\log n}$ (details of the proof are in Appendix~\ref{section:entropy-lower-bound-full-proof}). Later on, we will show how this transfers to a bound on the entropy.

\begin{lemma}
\label{lemma:lowerboundsupportsize}
If $\pi\colon \items \to [n]$ is chosen from a distribution of support size at most $k$, then any algorithm's probability of success against a worst-case adversary is at most $\frac{k + 1}{\log n}$.
\end{lemma}

To get a bound on the entropy, we show that for a low-entropy distribution there is a small subset of the support that is selected with high probability. More precisely, we have the following technical lemma whose proof can be found in Appendix~\ref{section:entropy-lower-bound-full-proof}.

\begin{lemma}
\label{lemma:entropyvssupportsize}
Let $a$ be drawn from a finite set $\mathcal{D}$ by a distribution of entropy $H$. Then for any $k \geq 4$ there is a set $T \subseteq \mathcal{D}$, $\lvert T \rvert \leq k$, such that $\Pr{a \in T} \geq 1 - \frac{8 H}{\log(k - 3)}$.
\end{lemma}

Finally, Theorem~\ref{theorem:lowerboundentropy} is proven as a combination of Lemma~\ref{lemma:lowerboundsupportsize} and Lemma~\ref{lemma:entropyvssupportsize}.

\begin{proof}[Proof of Theorem~\ref{theorem:lowerboundentropy}]
Set $k = \sqrt{\log n}$. Lemma~\ref{lemma:entropyvssupportsize} shows that there is a set of permutations $\Pi$ of size at least $k$ that is chosen with probability at least $1 - \frac{8 H}{\log(k - 3)}$. The distribution conditioned on $\pi$ being in $\Pi$ has support size only $k$. Lemma~\ref{lemma:lowerboundsupportsize} shows that if $\pi$ is chosen by a distribution of support size $k$, then the probability of success of any algorithm against a worst-case adversary is at most $\frac{k + 1}{\log n}$. Therefore, we get
\begin{align*}
\Pr{\text{success}} & = \Pr{\pi \in \Pi} \Pr{\text{success} \growingmid \pi \in \Pi} + \Pr{\pi \not\in \Pi} \Pr{\text{success} \growingmid \pi \not\in \Pi} \\
& \leq \Pr{\text{success} \growingmid \pi \in \Pi} + \Pr{\pi \not\in \Pi} \\
& \leq \frac{k + 1}{\log n} + \frac{8 H}{\log(k - 3)} \\
& = o(1) \enspace.
\end{align*}
\end{proof}

\section{Easy Distributions Are Hard to Characterize}
\label{sec:complexity}

Which distributions are s-admissible, meaning that they
allow an algorithm to achieve constant probability
of correct selection in the secretary problem?
The results in \secref{sec:non-uniform} and \secref{sec:entropy}
inspire hope that the $(k,\delta)$-UIOP, the $(p,q,\delta)$-BIP, 
or something very similar, is both
necessary and sufficient for s-admissibility.
Unfortunately, in this section we show that in some sense, it is
hopeless to try formulating 
a comprehensible condition that is both 
necessary and sufficient.
We construct a family of distributions $\natpd$ 
with associated algorithms $\algo$ having constant
success probability when the items are randomly
ordered according to $\natpd$, but the complicated
and unnatural structure of the distribution and 
algorithm underscore the pointlessness of precisely characterizing
s-admissible distributions. In more objective terms, 
we construct a $\natpd$ which is s-admissible,
yet for any algorithm whose stopping rule is computable
by circuits of  size less than $2^{n / \log(n)}$, 
the probability of correct selection is $o(1)$.

Throughout this section (and its corresponding
appendix) we will summarize the adversary's assignment
of values to items by a permutation $\advperm$;
the $j^{\mathrm{th}}$ largest value is assigned to item $\advperm(j)$.
If $\advpd$ is any probability distribution over such
permutations, we will let $\gval(\algo,\advpd)$ denote
the probability that $\algo$ makes a correct
selection when the adversary samples the value-to-item
assignment from $\advpd$, and nature independently samples
the item-to-time-slot assignment from $\natpd$. We will
also let 
\begin{align*}
\gval(\ast,\advpd) &= \max_{\algo} \gval(\algo,\advpd) \\
\gval(\algo,\ast) &= \min_{\advpd} \gval(\algo,\advpd) \\
\gval &=\min_{\advpd} \max_{\algo} \gval(\algo,\advpd).
\end{align*}
Thus, for example, the property that $\natpd$ is s-admissible
is expressed by the formula $\gval = \Omega(1)$.

As a preview of the techniques underlying
our construction, it is instructive to first
consider a game against nature in 
which there is no adversary, and the algorithm
is simply trying to pick out the maximum element
when items numbered in order of decreasing
value arrive in the random order specified
by $\natpd$. 
This amounts to determining $\gval(\ast,\idpd)$,
where $\idpd$ is the distribution that assigns probability 
1 to the identity permutation. Our construction is based
on the following intuition. In the secretary problem with
uniformly random arrival order, the arrival order of 
items that arrived before time $t$ is uncorrelated 
with the order in which items arrive after time $t$,
and so the ordering of past elements is irrelevant to
the question of whether to stop at time $t$. However,
there is a great deal of entropy in the ordering of 
elements that arrived before time $t$; it encodes
$\Theta(t \log t)$ bits of information. We will
construct a distribution $\natpd$ in which this
information contained in the ordering of the elements
that arrived before time $t=n/2$ fully encodes the
time when the maximum element will arrive after time $t$,
but in an ``encrypted'' way that cannot be decoded by
polynomial-sized circuits. We will make use of the 
well-known fact that a random function is hard on average
for circuits of subexponential size.

\begin{lemma} \label{lem:hard-on-avg}
If $\hardf : \{0,1\}^n \to [k]$ is a random function,
then with high probability there is 
no circuit of size $s(n) = 2^{n}/(8kn)$ that outputs the
function value correctly on more than $\frac{2}{k}$
fraction of inputs.
\end{lemma}

The simple proof of \lemref{lem:hard-on-avg} is included
in the appendix, for reference. 

\begin{theorem} \label{thm:game-against-nature}
There exists a family of distributions $\natpd \in \dists{S_n}$ 
such that $\gval(\ast,\idpd) = 1$, but for any algorithm $\algo$ whose
stopping rule can be computed by circuits of size
$s(n) = 2^{n/8}$, we have $\gval(\algo,\idpd) = O(1/n)$.
\end{theorem}
\begin{proof}
Assume for convenience that $n$ is divisible by 4. 
Fix a function $\hardf : \{0,1\}^{n/4} \to [n/2]$
such that no circuit of size $s(n) = 2^{\frac{n}{4}}/(n^2)$ outputs
the value of $\hardf$ correctly on more than $\frac{4}{n}$
fraction of inputs. The existence of such functions is 
ensured by Lemma~\ref{lem:hard-on-avg}. We use $\hardf$ to
define a permutation distribution $\natpd$ as follows.
For any binary string $x \in \{0,1\}^{n/4}$, define a 
permutation $\natperm(x)$ by performing the
following sequence of operations. First, rearrange
the items in order of increasing value by mapping
item $i$ to position $n-i+1$ for each $i$. Next, for
$i=1,\ldots,\frac{n}{4}$, swap the items in positions
$i$ and $i+\frac{n}{4}$ if and only if $x_i=1$. 
Finally, swap the items in positions
$n$ and $\frac{n}{2} + \hardf(x)$. (Note that this places
the maximum-value item in position $\frac{n}{2} + \hardf(x)$.)
The permutation distribution $\natpd$
is the uniform distribution over $\{\natperm(x) \mid
x \in \{0,1\}^{n/4}\}$.

It is easy to design an algorithm which always selects the
item of maximum value when the input sequence $\natperm$ is sampled
from $\natpd$. The algorithm first decodes the unique binary string
$x$ such that $\natperm = \natperm(x)$, by comparing the items
arriving at times $i$ and $i+\frac{n}{4}$ for each $i$ and setting the
bit $x_i$ according to the outcome of this comparison. Having
decoded $x$, we then compute $\hardf(x)$ and select the item
that arrives at time $\frac{n}{2} + \hardf(x)$. By construction,
when $\natperm$ is drawn from $\natpd$ this is always the element
of maximum value.

Finally, if $\algo$ is any secretary algorithm
we can attempt use $\algo$ to guess the value of $\hardf(x)$ for any
input $x \in \{0,1\}^{n/4}$ by the following simulation
procedure. First, define a permutation $\natperm'(x)$ by 
performing the same sequence of operations as in $\natperm(x)$
except for the final step of swapping the items in positions
$n$ and $n/2 + \hardf(x)$; note that this means that $\natperm'(x)$,
unlike $\natperm(x)$, can be constructed from input $x$ by a circuit
of polynomial size.
Now simulate $\algo$
on the input sequence $\natperm'(x)$, observe the
time $t$ when it selects an item, and output
$t - \frac{n}{2}$. The circuit complexity of this simulation
procedure is at most $\operatorname{poly}(n)$ times
the circuit complexity of the stopping rule implemented
by $\algo$, and the fraction of inputs $x$ on which it guesses
$\hardf(x)$ correctly is precisely $\gval(\algo,\idpd)$. 
(To verify this last statement, note that $\algo$  
makes its selection at time $t = \frac{n}{2} + \hardf(x)$ when
observing input sequence $\natperm(x)$ {\em if and only if}
if also makes its selection at time $t$ when observing
input sequence $\natperm'(x)$, because the two input 
sequences are indistinguishable to comparison-based
algorithms at that time.)
Hence,
if $\gval(\algo,\idpd) > \frac{4}{n}$ then the stopping rule
of $\algo$ cannot be implemented by circuits of size $2^{n/8}$.
\end{proof}

Our main theorem in this section derives essentially
the same result for the standard game-against-adversary
interpretation of the secretary problem, rather than
the game-against-nature interpretation adopted in 
\thmref{thm:game-against-nature}. 

\begin{theorem} \label{thm:complexity}
For any function $\ptz(n)$ such that 
$\lim_{n \to \infty} \ptz(n)=0$
while $\lim_{n \to \infty} \frac{n \cdot \ptz(n)}{\log n} = \infty$,
there exists a family of distributions $\natpd \in \dists{S_n}$ 
such that $\gval = \Omega(1)$, 
but any algorithm $\algo$ whose
stopping rule can be computed by circuits of size
$s(n) = 2^{n \, \ptz(n) / 4}$ satisfies $\gval(\algo,\ast) =  O(\ptz(n))$.
\end{theorem}
The full proof is provided in \appref{app:complexity}. Here we sketch
the main ideas.
\begin{proof}[Proof sketch.]
As in Theorem~\ref{thm:game-against-nature}, the algorithm and
``nature'' (\ie the process sampling the input order) 
will work in concert with each other to bring about a correct
selection, using a form of coordination that is 
information-theoretically easy but computationally hard. 
The difficulty lies in the fact that the adversary is
simultaneously working to thwart their efforts. If nature,
for example, wishes to use the first half of the input
sequence to ``encrypt'' the position where item 1 will
be located in the second half of the sequence,
then the adversary is free to assign the maximum value
to item 2 and a random value to item 1, rendering the
encrypted information useless to the algorithm.

Thus, our construction of the permutation distribution
$\natpd$ and algorithm $\algo$ will be guided by two
goals. First, we must ``tie the adversary's hands'' 
by ensuring that $\algo$ has constant probability of correct
selection unless the adversary's permutation, $\advperm$,
is in some sense ``close'' to the identity permutation.
Second, we must 
ensure that $\algo$
has constant probability of correct selection whenever
$\advperm$ is close to the identity, not only when it
is equal to the identity as in \thmref{thm:game-against-nature}. 
To accomplish the second goal we modify the 
construction in \thmref{thm:game-against-nature}
so that the first half of the input sequence encodes 
the binary string $x$ using an error-correcting code.
To accomplish the first goal we define $\natpd$ to be
a convex combination of two distributions: the
``encrypting'' distribution described earlier, and
an ``adversary-coercing'' distribution designed to
make it easy for the algorithm to select the maximum-value
element unless the adversary's permutation $\advperm$
is close to the identity in an appropriate sense.
\end{proof}

\section{Extensions Beyond Classic Secretary Problem}
\label{extension:sec}

We look at two generalizations of the classic secretary problem in this section, namely the \emph{multiple-choice secretary problem}, studied in~\citep{Kleinberg05}, and the \emph{online weighted bipartite matching problem},  studied extensively in~\citep{KorulaP09,KRTV13}, under our non-uniform permutation distributions. We give a positive result showing that a natural variant of the algorithm in~\cite{Kleinberg05} achieves a $(1-o(1))$-competitive ratio under our pseudo-random properties defined in \secref{sec:non-uniform}, while for the latter we show the algorithm proposed by~\cite{KorulaP09} fails to achieve any constant competitive ratio under our pseudo-random properties.

\paragraph*{Multiple-choice secretary problem}\label{section:k-unif}
We consider multiple-choice secretary problem (a.k.a. $k$-uniform matroid secretary problem). In this setting not only a single secretary has to be selected but up to $k$. An algorithm observes items with non-negative values based on the ordering $\p\colon \items\rightarrow \ts$ and chooses at most $k$ items in an online fashion. The goal is to maximize the sum of values of selected items. We consider distributions over permutations $\pi$ that fulfill the $(p,q,\delta)$-BIP, for some $p\geq k$. We show that a slight adaptation of the algorithm in \cite{Kleinberg05} achieves competitive ratio $1 - o(1)$, for large enough values of  $k$ and $q$ and small enough $\delta$.

The algorithm is defined recursively. We denote by $\alg{n'}{k'}{q'}$ the call of the algorithm that operates on the prefix of length $n'$ of the input. It is allowed to choose $k'$ items and expects $q'$ number of blocks. For $k' = 1$, $\alg{n'}{k'}{q'}$ is simply the standard secretary algorithm that we analyzed in Section~\ref{section:secretaryanalysis}. For $k' > 1$, the algorithm first draws a random number $\thresh(q')$ from a binomial distribution $\textrm{Binom}(q', \frac{1}{2})$ and then executes $\alg{\frac{\tau(q')}{q'}n'}{\lfloor k'/2 \rfloor}{\tau(q')}$. After round $\frac{\tau(q')}{q'}n'$ (we assume $n'$ is always a multiplier of $q'$), the algorithm accepts every item whose value is greater than the $\lfloor k'/2 \rfloor$-highest item arrived during rounds $1, \ldots, \frac{\tau(q')}{q'}n'$, until $k'$ items are selected by the algorithm or until round $n'$. Output is the union of all items returned by the recursive call and all items algorithm picked after the threshold round. 
We now have the following theorem, which is proved in Appendix~\ref{extension:sec}.

\begin{theorem}


\label{kunif:thm}
Suppose $\pi$ is drawn from a permutation distribution that satisfies  $(p,q,\delta)$-BIP for some $p\geq k$ and $\delta\leq \frac{1}{k^{\frac{1}{2}}}$. Then for all permutations $\sig$, $\alg{\items}{k}{q}$ is $(1 - O(\frac{1}{k^{\frac{1}{3}}})-\epsilon)$-competitive for the $k$-uniform matroid secretary problem, where $\epsilon$ can be arbitrary small for large enough value of $q$ and small enough value of $\delta$.
\end{theorem}

\paragraph*{Online weighted bipartite matching } Next, we consider online weighted bipartite matching, where the vertices on the offline side of a bipartite graph are given in advance and the vertices on the online side arrive online in a random order (not necessarily uniform). Whenever a vertex arrives, its adjacent edges with the corresponding weights are revealed and the online algorithm has to decide which of these edges should be included in the matching. The objective is to maximize the weight of the matching selected by online algorithm. A celebrated result of \citet{KorulaP09} shows the existence of a constant competitive online algorithm under uniform random order of arrival; nevertheless, this algorithm does not achieve any constant competitive ratio under our non-uniform assumptions for permutation distributions.

\begin{theorem}
\label{kp:thm}
For every $k$ and $\delta$, there is an instance and a probability distribution that fulfills the $(k, \delta)$-uniform-induced-ordering property such that the competitive ratio of the Korula-P\'{a}l algorithm is at least $\Omega\left(\frac{\delta^2}{(k+1)!} \frac{n}{\ln n}\right)$.
\end{theorem}


\section{Conclusion}

In this paper we have studied how secretary algorithms
perform when the arrival order satisfies relaxations 
of the uniform-random-order hypothesis.
We presented a pair of closely-related properties 
(the $(k,\delta)$-UIOP and the $(p,q,\delta)$-BIP)
that ensure that the standard secretary algorithm
has constant probability of correct selection, and
we derived some results on the minimum amount of entropy
and the minimum circuit complexity necessary to achieve
constant probability of correct selection in 
secretary problems with non-uniform arrival order.

We believe this work represents a first step toward
obtaining a deeper understanding of the amount and type
of randomness required to obtain strong performance guarantees
for online algorithms. The next step is to expand this study
beyond the setting of secretary problems. A very promising
domain for future investigation is online packing LP and its
generalization, online convex programming. Our positive 
result on the uniform matroid secrerary problem constitutes 
a first step toward obtaining a general positive result
confirming that existing algorithms such as the algorithms
of~\cite{KRTV14} and~\cite{AD15} preserve their performance
guarantees when the input ordering satisfies
$(k,\delta)$-UIOP or some other relaxation of the
uniform randomness assumption.
\newpage
\bibliographystyle{apalike}
\bibliography{secretary}
\appendix

%

\section{A secretary algorithm for $(3, \delta)$-induced-ordering property}
\label{app:3delta}
\begin{theorem} 
If a probability distribution fulfills the $(3, \delta)$-induced-ordering property, there is an algorithm for the secretary problem that selects the best item with probability $\frac{(1-\delta)^2}{6(1 + \delta)}$.
\end{theorem}

\begin{proof}
Consider the following algorithm: First we draw a threshold $\tau$ uniformly at random. Then we observe all items until round $\tau$. After round $\tau$, we accept the first item that is better than all items seen so far.

To analyze this algorithms let $x_1, x_2, \ldots, x_n$ be the items in order of decreasing value. To select $x_1$ it suffices that $x_2$ comes until round $\tau$ and $x_1$ comes after round $\tau$. For $i \geq 3$, let $Y_i$ be a 0/1 random variable indicating if $\pi(x_2) < \pi(x_i) < \pi(x_1)$.

Conditioned on $\sum_{i=3}^n Y_i = a$ and $\pi(x_2) < \pi(x_1)$, the probability that $x_2$ comes until round $\tau$ and $x_1$ comes after round $\tau$ is exactly $\frac{a+1}{n}$ because there are $a$ items coming between $x_2$ and $x_1$, giving $a + 1$ positive outcomes for $\tau$.

We have $\Ex{Y_i} \geq (1 - \delta) \frac{1}{3!} = \frac{1-\delta}{6}$. As $Y_i = 1$ implies $\pi(x_2) < \pi(x_1)$, we get $\Ex{Y_i \growingmid \pi(x_2) < \pi(x_1)} \geq \frac{1-\delta}{6 \Pr{\pi(x_2) < \pi(x_1)}} = \frac{1-\delta}{6 ( 1 - \Pr{\pi(x_2) > \pi(x_1)})} \geq \frac{1-\delta}{6 ( 1 - \frac{1-\delta}{2})} = \frac{1-\delta}{3( 1 + \delta)}$.

Overall, we get
\begin{align*}
\Pr{\text{select $x_1$} \growingmid \pi(x_2) < \pi(x_1)} & \geq \sum_{a=0}^{n-3} \Pr{\sum_{i=3}^n Y_i = a \growingmid \pi(x_2) < \pi(x_1)} \frac{a+1}{n} \\
& = \frac{1}{n} \left( 1 + \sum_{a=0}^{n-3} a \Pr{\sum_{i=3}^n Y_i = a \growingmid \pi(x_2) < \pi(x_1)} \right) \\
& = \frac{1}{n} \left( 1 + \Ex{\sum_{i=3}^n Y_i \growingmid \pi(x_2) < \pi(x_1)} \right) \\
& \geq \frac{1}{n} \left( 1 + \Ex{\sum_{i=3}^n Y_i \growingmid \pi(x_2) < \pi(x_1)} \right) \\
& \geq \frac{1}{n} + \frac{n-3}{n} \frac{1-\delta}{3( 1 + \delta)} \\
& \geq \frac{1-\delta}{3( 1 + \delta)} \enspace.
\end{align*}
Multiplying with $\Pr{\pi(x_2) < \pi(x_1)} \geq \frac{1-\delta}{2}$, we get
\[
\Pr{\text{select $x_1$}} \geq \frac{(1-\delta)^2}{6(1 + \delta)} \enspace.
\]
\end{proof}

\section{Deferred proofs}

\subsection{Proofs deferred from \secref{sec:non-uniform}}
\label{app:ipfull}
In this section we restate some of the results
from \secref{sec:non-uniform} and provide complete
proofs.

\subsubsection{Full Proof of Theorem~\ref{theorem:secretaryanalysis}}
The $(\indbound,\numblocks,\delta)$-block-independence property only makes statements about $\indbound$-tuples. We will need the bound also for smaller tuples. Indeed, using a simple counting argument we can show that this is already implicit in the definition.

\begin{lemma}
\label{lemma:block-ind-shorter-tuples}
If a distribution over permutations is $(\indbound,\numblocks,\delta)$-block-independent, then it is also $(\indbound',\numblocks,\delta)$-block-independent for any $\indbound' < \indbound$.
\end{lemma}

\begin{proof}
Given $x_1, \ldots, x_{\indbound'} \in \items$ and $b_1 \ldots, b_{\indbound'} \in [\numblocks]$, fill up the first tuple with arbitrary distinct entries $x_{\indbound' + 1}, \ldots, x_{p} \in \items$. The event $\bigwedge_{j \in [\indbound']} \pi^B(x_i) = b_i$ can now be expressed as the union of all events $\bigwedge_{j \in [\indbound]} \pi^B(x_i) = b_i$ over all tuples $(b_{\indbound' + 1} \ldots, b_{\indbound}) \in [\numblocks]^{\indbound - \indbound'}$. Note that these events are pairwise disjoint.  Therefore, the probability of their union is the sum of their probabilities, i.e.,
\begin{align*}
\Pr{\bigwedge_{j \in [\indbound']} \pi^B(a_i) = b_i} & = \Pr{\bigvee_{(b_{\indbound' + 1}, \ldots, b_\indbound) \in [\numblocks]^{\indbound - \indbound'}} \bigwedge_{j \in [\indbound]} \pi^B(x_i) = b_i} \\
& = \sum_{(b_{\indbound' + 1}, \ldots, b_\indbound) \in [\numblocks]^{\indbound - \indbound'}} \Pr{\bigwedge_{j \in [\indbound]} \pi^B(x_i) = b_i} \enspace.
\end{align*}
Using $(\indbound,\numblocks,\delta)$-block-independence and $\lvert[\numblocks]^{\indbound - \indbound'}\rvert = \numblocks^{\indbound - \indbound'}$, we get
\begin{align*}
\Pr{\bigwedge_{j \in [\indbound']} \pi^B(x_i) = b_i} & \geq \sum_{(b_{\indbound' + 1}, \ldots, b_p) \in [\numblocks]^{\indbound - \indbound'}} (1 - \delta) \left(\frac{1}{\numblocks}\right)^\indbound \\
& = \numblocks^{\indbound - \indbound'} (1 - \delta) \left(\frac{1}{\numblocks}\right)^\indbound \\
& = (1 - \delta) \left(\frac{1}{\numblocks}\right)^{\indbound'} \enspace.
\end{align*}
\end{proof}

\begin{proof}[Proof of Theorem~\ref{theorem:secretaryanalysis}]
Let $T = \lfloor \frac{\numblocks}{\e} \rfloor$ denote the index of the block in which the threshold is located. Furthermore, let $x_j \in \items$ be the $j$th best item. We condition on the event that $x_1$ comes in block with index $i$. To ensure that our algorithm picks this item, it suffices that $x_2$ comes in blocks $1, \ldots, T-1$. Alternatively, we also pick $x_1$ if the $x_2$ comes in blocks $i+1, \ldots, \numblocks$ and $x_3$ comes in blocks $1, \ldots, T-1$. Continuing this argument, we get
\[
\Pr{\text{correct selection}} \geq \sum_{i=T+1}^\numblocks \sum_{j=2}^\indbound \Pr{ \pi^B(x_1) = i, \pi^B(x_2), \ldots, \pi^B(x_{j - 1}) > i, \pi^B(x_j) < T } \enspace.
\]
We can now use Lemma~\ref{lemma:block-ind-shorter-tuples} and apply $(j,\numblocks,\delta)$-block-independence for $j \leq \indbound$. This gives us
\[
\Pr{\text{correct selection}} \geq \sum_{i=T+1}^\numblocks \sum_{j=2}^\indbound (1 - \delta) \frac{1}{\numblocks} \left( \frac{\numblocks - i}{\numblocks} \right)^{j-2} \frac{T - 1}{\numblocks} \enspace.
\]
We now reorder the sums and use the formula for finite geometric series. This gives us
\begin{align*}
\Pr{\text{correct selection}} & \geq (1 - \delta) \frac{T - 1}{\numblocks} \sum_{i=T+1}^\numblocks \frac{1}{\numblocks} \left( \sum_{j=2}^\indbound \left( \frac{\numblocks - i}{\numblocks} \right)^{j-2} \right) \\
& = (1 - \delta) \frac{T - 1}{\numblocks} \sum_{i=T+1}^\numblocks \frac{1}{\numblocks} \frac{ 1 - \left( \frac{\numblocks - i}{\numblocks} \right)^{\indbound-1} }{\frac{i}{\numblocks}} \\
& = (1 - \delta) \frac{T - 1}{\numblocks} \sum_{i=T+1}^\numblocks \frac{1}{i} \left(1 - \left( \frac{\numblocks - i}{\numblocks} \right)^{\indbound-1} \right) \\
& \geq (1 - \delta) \frac{T - 1}{\numblocks} \left(1 - \left( \frac{\numblocks - T}{\numblocks} \right)^{\indbound-1} \right) \sum_{i=T+1}^\numblocks \frac{1}{i} \enspace.
\end{align*}
We now apply the following bounds 
\[
\frac{T - 1}{\numblocks} \geq \frac{1}{\e} - \frac{2}{\numblocks} \enspace, \qquad \frac{\numblocks - T}{\numblocks} \leq 1 - \frac{1}{\e} \enspace, \qquad \text{ and }
\]
\[
\sum_{i=T+1}^\numblocks \frac{1}{i} \geq \int_{T+1}^{\numblocks+1} \frac{1}{x} dx = \ln\left( \frac{\numblocks + 1}{T + 1} \right) \geq \ln\left( \frac{\numblocks + 1}{\frac{\numblocks}{\e} + 1} \right) = 1 - \ln\left( \frac{\numblocks + \e}{\numblocks + 1} \right) \geq 1 - \left( \frac{\numblocks + \e}{\numblocks + 1} - 1\right) = 1 - \frac{\e - 1}{\numblocks + 1} \enspace.
\]
In combination, they imply
\begin{align*}
\Pr{\text{correct selection}} & \geq \left( \frac{1}{\e} - \frac{2}{\numblocks} \right) (1 - \delta) \left(1 - \left( 1 - \frac{1}{\e} \right)^{\indbound-1} \right)\left(1 - \frac{\e - 1}{\numblocks + 1}\right) \\
& \geq \left( \frac{1}{\e} - \frac{\e + 1}{\numblocks} \right) (1 - \delta) \left(1 - \left( 1 - \frac{1}{\e} \right)^{\indbound-1} \right)  \\
& \geq \frac{1}{\e} - \frac{\e + 1}{\numblocks} - \delta - \left( 1 - \frac{1}{\e} \right)^{\indbound-1} \enspace.
\end{align*}
\end{proof}

\subsubsection{Relation between the two properties}
\label{app:rbtp}

\begin{proof}[Proof of Theorem~\ref{thm:uiop-implies-bip}]
Note that it is safe to assume $\indbound \leq \numblocks$ as the statement is trivially fulfilled otherwise. Consider $\indbound$ distinct items $x_1, \ldots, x_\indbound \in \items$. To have $\pi(x_1) < \pi(x_2) < \ldots < \pi(x_\indbound)$, it suffices that these elements are mapped to different blocks and with an increasing sequence of indices. There are $\binom{\numblocks}{\indbound}$ such sequences. So, overall the probability is at least
\[
\binom{\numblocks}{\indbound} (1 - \delta) \left( \frac{1}{\numblocks} \right)^\indbound \geq \frac{(\numblocks-\indbound)^\indbound}{\indbound!} (1 - \delta) \left( \frac{1}{\numblocks} \right)^\indbound \geq \left(1 - \frac{\indbound}{\numblocks} \right)^\indbound (1 - \delta) \frac{1}{\indbound!} \geq \left(1 - \delta \frac{\indbound^2}{\numblocks} \right) \frac{1}{q!} \enspace.
\]
\end{proof}

\begin{proof}[Proof of Lemma~\ref{approxtheory:lemma1}]
We first define random variables $I_{i,j}\triangleq \mathbf{I}_{\pi(i)\leq \pi(j)}$ and $\tilde{I}_{i,j}\triangleq \mathbf{I}_{\phi(i)\leq \phi(j)}$ for all $i,j\in [n]$. Note that for all $i\in[p]$, $X_i=\pi(x_i)/n=\frac{\sum_{j=1}^n I_{i,j}}{n}$. This implies that
\begin{align}
\Ex{\prod_{i=1}^p X_i^{k_i}}=\Ex{\prod_{i=1}^p {\left(\frac{\sum_{j=1}^n I_{i,j}}{n}\right)}^{k_i}}=\frac{\Ex{\prod_{i=1}^p {\left(\sum_{j=1}^n I_{i,j}\right)}^{k_i}}}{n^{(\sum_{i=1}^p{k_i})}}
 \end{align}
 By expanding the numerator due to linearity of expectation, we will have a sum of expectation of algebraic terms in the numerator, where each algebraic term multiplication of at most $\frac{k}{2p}\times p=k/2$ indicators $I_{i,j}$. There are at most $2\times k/2=k$ particular items involved in these indicator functions. Now, lets look at one of the terms, e.g. $\Ex{\prod_{l=1}^{k/2} I_{i_l,j_l}}$ in which $k$ items $\{x_{s_1},\ldots,x_{s_k}\}$ are involved.  The product $\prod_{l=1}^k I_{i_l,j_l}$ forces the induced ordering of elements $\{x_{s_1},\ldots,x_{s_k}\}$ be in a particular subset $S\subseteq S_k$. 
 
 Hence, $\Ex{\prod_{l=1}^{k/2} I_{i_l,j_l}}=\Pr{\textrm{induced ordering by $\pi$ over $\{x_{s_1},\ldots,x_{s_k}\}$ will be in $S$}}$. Now as $\pi$ satisfies the $(k,\delta)$-uniform-induced-ordering property, we have
 \begin{align}
 \Ex{\prod_{l=1}^{k/2} I_{i_l,j_l}}&=\Pr{\textrm{induced ordering by $\pi$ over $\{x_{s_1},\ldots,x_{s_k}\}$ will be in $S$}}\nonumber\\
 &\geq (1-\delta) \Pr{\textrm{induced ordering by $\phi$ over $\{x_{s_1},\ldots,x_{s_k}\}$ will be in $S$}}\nonumber\\
 &\label{approxtheory:eq1}= (1-\delta) \Ex{\prod_{l=1}^{k/2} \tilde{I}_{i_l,j_l}}
 \end{align}
 From (\ref{approxtheory:eq1}) one can conclude that 
 \begin{align}
\Ex{\prod_{i=1}^p X_i^{k_i}}\geq (1-\delta)\Ex{\prod_{i=1}^p {\left(\frac{\sum_{j=1}^n \tilde{I}_{i,j}}{n}\right)}^{k_i}}=(1-\delta)\Ex{\prod_{i=1}^p \left(\frac{n\phi(i)}{n}\right)^{k_i}}=(1-\delta)\Ex{\prod_{i=1}^p {\phi(i)}^{k_i}}
 \end{align}
 which completes the proof.
\end{proof}
\begin{proof}[Proof of Lemma\ref{approxtheory:lemma2}]
We define continuous functions $f_i\colon [0, 1] \to \RR$ for $i\in [p]$ by
\[
f_i(x) = \begin{cases}
0 & \text{ for $x < a_i$ or $x > b_i$} \\
\frac{x - a_i}{\gamma} & \text{ for $a_i \leq x \leq a_i + \gamma$} \\
- \frac{x - b_i}{\gamma} & \text{ for $b_i - \gamma \leq x \leq b_i$} \\
1 & \text{ for $a_i + \gamma \leq x \leq b_i - \gamma$}
\end{cases}
\]
where $\lvert b_i-a_i \rvert \geq 2\gamma$. Note that all of these functions are continuous and satisfy condition of Theorem~\ref{lemma:approximationtheory} for $\omega_{f_i}(x)=\frac{x}{\gamma}$.


Observe that $f_i$ is point-wise smaller than the indicator function $\mathbf{1}_{[a_i, b_i]}$ that has value $1$ between $a_i$ and $b_i$ and $0$ otherwise. Therefore, we have $\Pr{\bigwedge_{i=1}^{p}\left(X_i\in [a_i,b_i]\right)} = \Ex{\prod_{i=1}^p\mathbf{1}_{[a_i, b_i]}(X_i)} \geq \Ex{\prod_{i=1}^pf_i(X_i)}$.

By Theorem~\ref{lemma:approximationtheory}, for every $i$ there is a polynomial function $g_i\colon [0, 1] \to \RR$ of degree $d$ such that $\lVert f_i- g_i \rVert_{\infty} \leq \frac{3}{2}\omega_{f_i}(\frac{1}{\sqrt d})\leq \frac{3}{2\gamma\sqrt d}$. We now have $g_i(\phi(i)) \geq f_i(\phi(i)) -\frac{3}{2\gamma\sqrt d} \geq \mathbf{1}_{[a_i+\gamma, b_i-\gamma]}(\phi(i)) - \frac{3}{2\gamma\sqrt d} $ and therefore 
\begin{align}
\Ex{\prod_{i=1}^p g_i(\phi(i))} &\geq \Ex{\prod_{i=1}^p\left(\mathbf{1}_{[a_i+\gamma, b_i-\gamma]}(\phi(i)) - \frac{3}{2\gamma\sqrt d}\right) } \geq \prod_{i=1}^{p}\Pr{\phi(i)\in [a_i+\gamma,b_i-\gamma]}-\frac{3p}{2\gamma\sqrt d}\nonumber\\
&= \prod_{i=1}^p (b_i-a_i-2\gamma)-\frac{3p}{2\gamma\sqrt d}\geq \prod_{i=1}^p (b_i-a_i)-(2p\gamma+\frac{3p}{2\gamma\sqrt d})\overset{(1)}{\geq}\prod_{i=1}^p (b_i-a_i)-\frac{4p}{d^{1\!/\!4}}
\end{align}
where to get inequality (1) we set $\gamma=\frac{1}{2}d^{-\frac{1}{4}}$. Furthermore, as $g_i$ is a polynomial function of degree at most $d$ and $\Ex{\prod_{i=1}^p X_i^{k_i}} \geq\Ex{\prod_{i=1}^p {\phi(i)}^{k_i}}(1-\delta)$ for all $k_i\leq d$, we get by linearity of expectation $\Ex{\prod_{i=1}^{p}g(X_i)} \geq (1-\delta)\Ex{\prod_{i=1}^p g_i(\phi(i))}$. Now we use $g_i(x) \leq f_i(x) +\frac{3}{2\gamma\sqrt d}=\frac{3}{d^{1\!/\!4}}$, giving us
\begin{align}
\Ex{\prod_{i=1}^pf_i(X_i)}& \geq \Ex{\prod_{i=1}^p (g_i(X_i)-\frac{3}{d^{1\!/\!4}})}\geq \Ex{\prod_{i=1}^p g_i(X_i)}-\frac{3p}{d^{1\!/\!4}}\geq (1-\delta)\Ex{\prod_{i=1}^p g_i(\phi(i))}-\frac{3p}{d^{1\!/\!4}}\nonumber\\
&\geq (1-\delta)\left(\prod_{i=1}^p [b_i-a_i]-\frac{4p}{d^{1\!/\!4}}\right)-\frac{3p}{d^{1\!/\!4}}\geq  (1-\delta)\left(\prod_{i=1}^p [b_i-a_i]\right)-\frac{7p}{d^{1\!/\!4}}
\end{align}

Overall, we get $\Pr{\bigwedge_{i=1}^{p}\left(X_i\in [a_i,b_i]\right)} \geq  (1-\delta)\left(\prod_{i=1}^p (b_i-a_i)\right)-\frac{7p}{d^{1\!/\!4}}$, as desired.
\end{proof}

\subsubsection{Full proof of \thmref{thm:rip}}
\label{app:rip}

\begin{proof}
For any $k$-tuple of indices $(i_1,\ldots,i_k)$
we must show that each
of the $k!$ possible orderings of $(w \cdot i_1),\ldots,(w \cdot i_k)$
has probability at least $\frac{1-\delta}{k!}$. By 
symmetry it suffices to show that the probability of the
event $\{ w \cdot x_1 < w \cdot x_2 < \cdots < w \cdot x_k \}$
is at least $\frac{1-\delta}{k!}$. This event is unchanged by
rescaling $w$, so we are free to substitute whatever spherically-symmetric
distribution we wish. Henceforth assume $w$ is sampled from the
multivariate normal distribution $\mathcal{N}(0,1)$, whose
density function is $(2 \pi)^{-d/2} \exp \big( \frac12 \|w\|^2 \big)$.

Let $X_k$ denote the 
matrix whose $k$ columns are the vectors $x_1,\ldots,x_k$
and let $A = X_k^\trans$ denote its transpose. 
Scaling $x_1,\ldots,x_k$ by a common scalar, if 
necessary, we are free to assume that $\det(X_k^\trans X_k)=1$.
The RIP implies that the ratio of the largest and smallest
right singular values of $X_k$ is at most $\frac{1+\delta_k}{1-\delta_k}$,
and since their product is 1 this means that the
smallest singular value is at least $\frac{1-\delta_k}{1+\delta_k}$.

Now let
$\cone = \{ z \in \RR^k \mid z_1 < z_2 < \cdots < z_k \}$. The
 event $\{ w \cdot x_1 < w \cdot x_2 < \cdots < w \cdot x_k \}$
can be expressed more succinctly as $\{ Aw \in \cone \}$, and its
probability is
\[
\pr{ Aw \in \cone } = \int_{w \in A^{-1}(\cone)} (2 \pi)^{-d/2} 
\exp \big( -\tfrac12 \|w\|^2 \big) \, dw.
\]
The matrix $A$ is not square, hence not invertible; the notation
$A^{-1}(\cone)$ merely means the inverse-image of $\cone$ under
the linear transformation $\RR^d \to \RR^k$ represented by $A$.
The Moore-Penrose pseudoinverse of $A$ is the matrix
$X_k (X_k^{\trans} X_k)^{-1}$, which we denote henceforth
by $A^+$.
We can write any $w \in A^{-1}(\cone)$ uniquely as 
$z + A^+ y$, where $y \in \cone$, $z \in \ker(A)$, and $z$ is
orthogonal to $A^+ y$. By our scaling assumption, the product
of the singular values of $A^+$ equals 1, which justifies
the second line in the following calculation
\begin{align*}
\pr{ Aw \in \cone } &= \int_{w \in A^{-1}(\cone)} (2 \pi)^{-d/2} 
\exp \big( -\tfrac12 \|w\|^2 \big) \, dw \\
 & = \int_{y \in \cone} \int_{z \in \ker(A) }
(2 \pi)^{-d/2}  \exp \big( - \tfrac12 \|z\|^2 - \tfrac12 \| A^+ y \|^2 \big)
 \, dz \, dy \\
  & = 
\left( \int_{y \in \cone} (2 \pi)^{-k/2} \exp \big( -\tfrac12 \|A^+ y \|^2 \big)
\, dy \right) 
\;
\left( \int_{z \in \ker(A)} (2 \pi)^{-(d-k)/2} \exp \big( -\tfrac12 \|z\|^2 \big)
\, dz \right) \\
 & = 
\left( \int_{y \in \cone} (2 \pi)^{-k/2} \exp \big( -\tfrac12 \|A^+ y \|^2 \big)
\, dy \right).
\end{align*}
We can rewrite the right side as an integral in spherical coordinates.
Let $d \omega$ denote the volume element of the unit sphere $S^{k-1}
\subset \RR^k$ and
let $\conesph = \cone \cap S^{k-1}$. Then writing $y = r u$, where
$r \geq 0$ and $u$ is a unit vector, we have
\begin{align*}
\pr{ Aw \in \cone } & = \int_{u \in \conesph} \int_{r=0}^{\infty}
 (2 \pi)^{-k/2} \exp \big( -\tfrac12 r^2 \|A^+ u \|^2 \big) r^{k-1}
\, dr \, d\omega(u) \\
 & =
\int_{u \in \conesph} (2 \pi)^{-k/2} \|A^+ u\|^{-k}
\int_{s=0}^{\infty} \exp \big( -\tfrac12 s^2 \big) s^{k-1} \, ds  \, d\omega(u).
\end{align*}
The singular values of $A^+$ are the multiplicative inverses of the
singular values of $X_k$, hence the largest singular value of $A^+$
is at most $\frac{1+\delta_k}{1-\delta_k}$. In other words,
$\|A^+ u\| \leq \frac{1+\delta_k}{1-\delta_k}$ for any unit vector $u$.
Plugging this bound into the integral above, we find that
\[
\pr{ Aw \in \cone } \geq 
\left( \tfrac{1-\delta_k}{1+\delta_k} \right)^k 
\int_{u \in \conesph} (2 \pi)^{-k/2}
\int_{s=0}^{\infty} \exp \big( -\tfrac12 s^2 \big) s^{k-1} \, ds  \, d\omega(u)
 = 
\frac{1}{k!} \left( \tfrac{1-\delta_k}{1+\delta_k} \right)^k ,
\]
where the last equation is derived by observing that the integral is
equal to the Gaussian measure of $\cone$. Finally, by our choice of
$\delta_k$, we have $\big( \frac{1-\delta_k}{1+\delta_k} \big)^k > 1-\delta$,
which concludes the proof.
\end{proof}

\subsubsection{Full proof of Theorem~\ref{thm:rlec}}
\label{app:small-support}


\begin{proof}[]
We show this claim using the probabilistic method. Permutation $\pi_i\colon \items \to [n]$ is drawn uniformly at random from the set of all permutations with replacement. We claim that the set $S = \{ \pi_1, \ldots, \pi_\xi \}$ fulfills the stated condition with probability at least $1 - \frac{1}{n}$.

Fix $k$ distinct items $x_1, \ldots, x_k \in \items$. Let $Y_i = 1$ if $\pi_i(x_1) < \pi_i(x_2) < \ldots < \pi_i(x_k)$. As $\pi_i$ is drawn uniformly from the set of all permutations, we have $\Pr{Y_i = 1} = \frac{1}{k!}$. That is, we have $\Ex{\sum_{i = 1}^\xi Y_i} = \frac{\xi}{k!}$. As the random variables $Y_i$ are independent, we can apply a Chernoff bound. This gives us
\[
\Pr{\sum_{i = 1}^\xi Y_i \leq (1 - \delta) \frac{\xi}{k!}} \leq \exp\left( - \frac{\delta^2}{2} \frac{\xi}{k!} \right) = n^{k+1} \enspace.
\]
Note that if $\sum_{i = 1}^\xi Y_i \leq (1 - \delta) \frac{\xi}{k!}$ then the respective sequence $x_1, \ldots, x_k \in \items$ has probability at least $(1 - \delta) \frac{1}{k!}$ when drawing one permutation at random from $S$.

There are fewer than $n^k$ possible sequences. Therefore, applying a union bound, with probability at least $1 - \frac{1}{n}$ the bound is fulfilled for all sequences simultaneously and therefore $S$ fulfills the stated condition.
\end{proof}

\subsection{Proofs deferred from \secref{sec:entropy}}
\label{section:entropy-lower-bound-full-proof}
In this section we provide complete
proofs of some of the results in \secref{sec:entropy}.

\begin{proof}[Proof of Lemma~\ref{lemma:lowerboundsupportsize}]
Let $\Pi$, $\lvert \Pi \rvert \leq k$, be the support of the distribution $\pi$ is drawn from. Lemma~\ref{low1:lemma-short} shows that there is a sequence $(x_1, \ldots, x_s)$ of length $s = \frac{\log n}{k + 1}$ that is semitone with respect to any permutation in $\pi$.

It only remains to apply Yao's principle: Instead of considering the performance of a random $\pi$ against a deterministic adversary, we consider the performance of a fixed $\pi$ against a randomized adversary. Lemma~\ref{low3:lemma} shows that there is a distribution over instances such that no $\pi \in \Pi$ has success probability better than $\frac{1}{s} = \frac{\log n}{k + 1}$. 
\end{proof}

\begin{proof}[Proof of Lemma~\ref{lemma:entropyvssupportsize}]
Set $\alpha = \frac{H}{\log(k - 3)}$ and $\beta = \frac{\alpha}{k - 3}$. Note that for $\alpha \geq \frac{1}{8}$, the statement becomes trivial. Therefore, we can assume without loss of generality that $\alpha < \frac{1}{8}$. This implies $\log(\alpha) < 0$. Therefore, we get
\[
\frac{H}{- \log \beta} = \frac{\alpha \log(k - 3)}{\log(k - 3) - \log(\alpha)} \leq \alpha \enspace.
\]

Let $a_1, \ldots, a_k$ be the elements of $\mathcal{D}$ such that $p_{a_i} \geq \beta$ for all $i$ and $p_{a_1} \geq p_{a_2} \geq \ldots \geq p_{a_k}$. Furthermore, partition $\mathcal{D} \setminus \{ a_1, \ldots, a_k \}$ into $S_1, \ldots, S_\ell$ such that $p_{S_i} \in [\beta, 2 \beta)$ for $i < \ell$, $p_{S_\ell} < 2 \beta$

Observe that $p_{a_3} \leq \frac{1}{\e}$ because probabilities sum up to at most $1$. Therefore, for $i \geq 3$, we have $-p_{a_i} \log(p_{a_i}) \geq - \beta \log \beta$ by monotonicity. Furthermore, for all $j < \ell$, we have $-p_{S_j} \log(p_{S_j}) \geq - \beta \log \beta$. In combination, this gives us
\[
H \geq \sum_{i=3}^k -p_{a_i} \log(p_{a_i}) + \sum_{j=1}^{\ell - 1} -p_{S_i} \log(p_{S_i}) \geq (k + \ell - 3)(- \beta \log \beta) \enspace.
\]

For $k$ and $\ell$, this implies
\[
k \leq \frac{H}{- \beta \log \beta} + 3 \leq \frac{\alpha}{\beta} + 3 \leq k \quad \text{ and } \quad \ell \leq \frac{H}{- \beta \log \beta} + 3 \leq \frac{\alpha}{\beta} + 3 \enspace.
\]
In conclusion, we have
\[
\sum_{j = 1}^\ell p_{S_j} \leq 2 \beta \ell \leq 2 \alpha + 6 \beta \leq 8 \alpha \enspace.
\]
\end{proof}

\subsection{Proofs deferred from \secref{sec:complexity}}
\label{app:complexity}

In this section we restate some of the results
from \secref{sec:complexity} and provide complete
proofs.

\begin{lemma} \label{lem:hardf}
If $\hardf : \{0,1\}^n \to [k]$ is a random function,
then with high probability there is 
no circuit of size $s(n) = 2^{n}/(8kn)$ that outputs the
function value correctly on more than $\frac{2}{k}$
fraction of inputs.
\end{lemma}
\begin{proof}
The proof closely parallels the proof of the corresponding
statement for worst-case hardness rather than hardness-on-average,
which is presented, for example, in the textbook by \citet{AroraBarak}.
The number of Boolean circuits of size $s$ is bounded
by $s^{3s}$. For any one of these circuits, $C$, the expected
number of inputs $x$ such that $C(x)=\hardf(x)$ is 
$\frac1k \cdot 2^n$. Since the events $\{C(x)=\hardf(x)\}$
are mutually independent as $x$ varies over $\{0,1\}^n$, 
the Chernoff bound (e.g.,~\cite{MR95}) implies that the probability of  
more than $\frac2k \cdot 2^n$ of these events taking place is
less than $\exp \big( - \frac{1}{3k} \cdot 2^n \big)$.
The union bound now implies that the probability there exists
a circuit $C$ of size $s$ that correctly computes $\hardf$ on
more than $\frac2k$ fraction of inputs is bounded above by
$
\exp \big( 3s \ln(s) - \frac{1}{3k} \cdot 2^n \big).
$
When $s = 2^n/(8kn)$ this yields the stated high-probability bound.
\end{proof}

In the sequel we will need a version of the lemma
above in which the circuit, rather than being constructed from the usual
Boolean gates, is constructed from $t$ different types of gates, 
each having $m$ binary inputs and one binary output.

\begin{lemma} \label{lem:very-hardf}
Suppose we are given $t$ types of gates, each computing a 
specific function from $\{0,1\}^m$ to $\{0,1\}$. 
If $\hardf : \{0,1\}^n \to [k]$ is a random function,
then with high probability there is 
no circuit of size 
$s(n) \leq 2^n / (8k \cdot \max\{mn, \, \ln(t)\})$
that outputs the
function value correctly on more than $\frac{2}{k}$
fraction of inputs.
\end{lemma}
\begin{proof}
The proof is the same except that the number of circuits,
rather than being bounded by $s^{3s}$, is now bounded
by $(t s^m)^s = \exp ( s \ln t + ms \ln s )$.
The stated high-probability bound continues to hold
if $m s \ln s < \frac{1}{8k} \cdot 2^n$ and 
$s \ln t < \frac{1}{8k} \cdot 2^n$.
The assumption 
$s(n) \leq 2^n / (8k \cdot \max\{mn, \, \ln(t)\})$
justifies these two inequalities and completes the proof.
\end{proof}

\subsubsection{Proof of \thmref{thm:complexity}}
\label{sec:complexity-proof}

Similar to the proof of \thmref{thm:game-against-nature},
our plan for proving \thmref{thm:complexity} is to 
construct a distribution over arrival orderings, $\natpd$,
in which the first half of the input sequence attempts to
encode the position where the maximum-value item occurs
in the second half of the permutation.
What makes the proof more difficult is
that the adversary chooses the 
ordering of items by value (as represented by
a permutation $\advperm \in S_n$), and this
ordering could potentially be chosen to thwart the
decoding process. In our construction we will make 
a distinction between {\em decodable} orderings---whose
properties will guarantee that our decoding algorithm succeeds
in finding the maximum-value item---and {\em non-decodable}
orderings, which may lead the decoding algorithm to make
an error. We will then design a separate algorithm that 
succeeds with constant probability when the adversary's ordering
is non-decodable.

We will assume throughout the proof that $n$ is divisible
by 8, for convenience. Recall, also, that the theorem
statement declares $\ptz(n)$ to be any 
function of $n$ such that $\lim_{n \to \infty} \ptz(n)=0$
while $\lim_{n \to \infty} \frac{n \cdot \ptz(n)}{\log n} = \infty$.
For convenience we adopt the 
notation $[a,b]$ to denote the subset of $[n]$
consisting of integers in the range from $a$ to $b$, 
inclusive; analogously, we may refer to subsets
of $[n]$ using open or half-open interval notation.
\begin{definition} \label{def:decodable}
A total ordering of $[n]$ is {\em decodable} if it
satisfies the following properties.
\begin{enumerate}
\item The maximal element of the ordering is $n$.
\item \label{d2-property}
When the elements of the set
$ ( \tfrac{n}{4}, \, \tfrac{n}{2} ] $
are written in decreasing order with respect to the total ordering,
the
of the first $n \cdot \ptz(n)$ elements 
that belong to $ ( \tfrac{3n}{8}, \, \tfrac{n}{2} ]$
is at least $\frac{39}{40}$.
\end{enumerate}
\end{definition}
Our proof will involve the construction of three
distributions over permutations, and three 
corresponding algorithms.
\begin{itemize}
\item an ``encrypting'' distribution $\npde$ that hides item $n$ in the 
second half of the permutation while arranging the first
half of the permutation to form a ``clue''
that reveals the location of item $n$, 
but does so in a way that 
cannot be decrypted by small circuits;
\item a first ``adversary-coercing'' distribution $\npdacf$
that forces the adversary to make item $n$ the
most valuable item;
\item a second ``adversary-coercing'' distribution $\npdacs$
that forces the adversary to satisfy the second property
in the definition of a decodable ordering.
\end{itemize}
Corresponding to these three distributions we will define
algorithms $\algoe,\algoacf,\algoacs$ such that:
\begin{itemize}
\item $\algoe$ has constant probability of correct selection
when the adversary chooses a decodable ordering;
\item $\algoacf$ has constant probability of correct selection
when the adversary chooses an ordering that violates the first
property of decodable orderings;
\item $\algoacs$ has constant probability of correct selection
when the adversary chooses an ordering that satisfies the first
property of decodable orderings but violates the second.
\end{itemize}
Combining these three statements, one can easily conclude
that when
nature samples the arrival order using the
permutation distribution $\natpd = 
\frac13 (\npde + \npdacf + \npdacs)$, and 
when the algorithm $\algo$ is the one that randomizes
among the three algorithms $\{\algoe,\algoacf,\algoacs\}$ 
with equal probability, then $\algo$ has constant 
probability of correct selection no matter how the
adversary orders items by value.

We begin with the construction of the distribution
$\npdacf$ and algorithm $\algoacf$. The following
describes the procedure of drawing a random sample
from $\npdacf$.
\begin{algorithm}[H]
\caption{Sampling procedure for $\npdacf$}
\begin{algorithmic}[1]
\State Sample an $\big( \frac{n}{2} \big)$-element
set $L \subset [n-1]$ uniformly at random.
\State Let $R = [n-1] \setminus L$.
\State Let $\natperm'$ by the permutation that lists the
elements of $L$ in increasing order,
followed by the elements of $R$ 
increasing order, followed by $n$.
\label{step:npdacf3}
\State
Choose a uniformly random $i \in \big[ 1, \, \frac{n}{2} \big]$
and let $\tau_i$ be the transposition that swaps elements $n$ and $i$.
(If $i=n$ then $\tau_i$ is the identity permutation.)
\State Let $\natperm = \tau_i \circ \natperm'$.
\end{algorithmic}
\end{algorithm}
Define $\algoacf$ to be an algorithm that observes the
first $\frac{n}{2}$ elements, sets a threshold equal to
the maximum of the observed elements, and selects the next element
whose value exceeds this threshold. In the following lemma
and for the remainder of this section, $\revperm$ denotes
the permutation
that lists the items
in order of increasing value, i.e., $\revperm(i) = n-i$.

\begin{lemma} \label{lem:npdacf}
If the adversary's ordering $\advperm$ assigns the maximum value
to any item other than $n$, then $\gval[\npdacf](\algoacf,\advperm) >
\frac{1}{4}$. On the other hand, 
$\gval[\npdacf](\ast,\revperm) = \frac{2}{n}.$
\end{lemma}
\begin{proof}
Suppose that $\advperm$ assigns the maximum value to item $i \neq n$,
and suppose that item $j$ receives the second-largest value among
items in $[n-1]$. In the sampling procedure for $\npdacf$, the 
event that $j \in L$ and $i \in R$ has probability
$$
\frac{n/2}{n-1} \, \cdot \, \frac{(n/2)-1}{n-2} > \frac{1}{4},
$$
and when this event happens the algorithm $\algoacf$ is guaranteed
to select item $i$.

To prove the second part of the lemma, suppose the adversary
assigns values to items in increasing order and observe that
this guarantees that the first $n/2$ items in the 
permutation $\natperm'$ (defined in Step~\ref{step:npdacf3}
of the sampling procedure) are listed in increasing order
of value, and that the first $n/2$ items in $\natperm$ are
the same except that the value at index $i$ is replaced
by the maximum value. Now consider any algorithm and let 
$t$ denote the time when it makes its selection when facing
a monotonically increasing sequence of $n$ values. If $t > n/2$,
then the algorithm assuredly makes an incorrect selection when
facing the input sequence $\natperm \revperm$. If $t \leq n/2$,
then the algorithm makes a correct selection if and only if 
$t$ matches 
the random index $i$ chosen in the sampling procedure for
$\natperm$, an event with probability $2/n$.
\end{proof}

We next present the construction of $\npdacs$.
\begin{algorithm}[H]
\caption{Sampling procedure for $\npdacs$}
\begin{algorithmic}[1]
\State \label{step:acs1}
With probability $\frac12$, reverse the order
of the first $\frac{n}{4}$ items in the list.
\State Initialize $I = \emptyset$.
\For{$i = 1,\ldots,\frac{n}{4}$}
\State With probability $\frac{1}{n \, \ptz(n)}$:
\State \hspace*{5em} Swap
the items in positions $i$ and $i + \frac{n}{4}$.
\State \hspace*{5em} Add $i$ to the set $I$.
\State \hspace*{5em} If $i > \frac{n}{8}$ then add $i$ into the set $I^+$.
\EndFor
\State \label{step:acs3}
Let $\natperm'$ denote the permutation of items defined
at this point in the procedure.
\If{$I^+$ is non-empty}
\State \label{step:acs4}
Choose a uniformly random index $i \in I^+$
\Else
\State Choose a uniformly random index $i \in \big( \frac{n}{8}, 
\frac{n}{4} \big]$.
\EndIf
\State Let $\tau_i$ be the transposition that swaps elements $n$ and $i$.
\State Let $\natperm = \tau_i \circ \natperm'$.
\end{algorithmic}
\end{algorithm}
Define $\algoacs$ to be an algorithm that observes the
first $\frac{n}{8}$ elements, sets a threshold equal to
the maximum of the observed elements, and selects the next element
whose value exceeds this threshold.

\begin{lemma} \label{lem:npdacs}
If the adversary's ordering $\advperm$ assigns the maximum value
to item $n$ but violates Property~\ref{d2-property} in the 
definition of a decodable ordering, then 
$\gval[\npdacs](\algoacs,\advperm) >
\frac{1}{250}$. 
On the other hand, 
$\gval[\npdacf](\ast,\revperm) =
O(\ptz(n))$.
\end{lemma}
\begin{proof}
First suppose that $\advperm$ assigns the maximum value
to item $n$ but violates Property~\ref{d2-property} in the 
definition of a decodable ordering. To prove that
$\gval[\npdacs](\algoacs,\advperm) >
\frac{1}{250}$, note first that $\algoacs$ is guaranteed
to make a correct selection if the permutation $\natperm'$
(defined in step~\ref{step:acs3} of the sampling procedure)
has the property that the highest-value item found among
the first $n/4$ positions in $\natperm'$
belongs to one of the first $\frac{n}{8}$ positions. 
Recalling the set $I$ defined in the sampling procedure,
let $J = \big[ 1, \frac{n}{4} \big] \setminus I$ and
let $K = \big\{ i + \frac{n}{4} \mid i \in I \big\}$.
Note that $J \cup K$ is the set of items found among the
first $n/4$ positions in $\natperm'$. Let $j,k$ denote the
highest-value elements of $J$ and $K$, respectively. 
(If $K$ is empty then $k$ is undefined.)
Step~\ref{step:acs1} of the sampling procedure ensures
that with probability $\frac12$, item $j$ belongs to one
of the first $\frac{n}{8}$ positions, and this event is
independent of the event that $K$ is non-empty and 
$k$ belongs to one of the 
first $\frac{n}{8}$ positions. To complete the proof, we
now bound the probability of that event from below by $\frac{1}{100}$.

Let $i_1,i_2,\ldots,i_{n/4}$ denote a listing of the 
elements of the set $\big( \frac{n}{4}, \frac{n}{2} \big]$
in decreasing order of value. For $1 \leq \ell \leq \frac{n}{4}$, the 
probability that $k = i_\ell$ is 
$\big(1 - \frac{1}{n \, \ptz(n)} \big)^{\ell-1} \frac{1}{n \, \ptz(n)}$.
Let $L$ denote the set of $\ell \leq n \, \ptz(n)$ 
such that $i_\ell \leq \frac{3n}{8}$. By our hypothesis
that $\advperm$ violates Property~\ref{d2-property} in the 
definition of a decodable ordering, we know that 
$\frac{|L|}{n \, \ptz(n)} > \frac{1}{40}$.
If $k = i_\ell$ for any $\ell \in L$, then $k$ belongs to one
of the first $\frac{n}{8}$ positions in $\natperm'$. The 
probability of this event is 
\[
\sum_{\ell \in L} \Pr{k = i_\ell} = 
\sum_{\ell \in L} \left(
1 - \frac{1}{n \, \ptz(n)} \right)^{\ell-1} \frac{1}{n \, \ptz(n)} 
\geq
\frac{1}{n \, \ptz(n)} \sum_{\ell \in L} \left(
1 - \frac{1}{n \, \ptz(n)} \right)^{n \, \ptz(n) - 1}
\geq \frac{|L|}{n \, \ptz(n)} \cdot \frac{1}{e} > \frac{1}{100},
\]
as desired.

The second half of the lemma asserts that 
$\gval[\npdacf](\ast,\revperm) = O(\ptz(n))$,
where $\revperm$ denotes the permutation
that lists the items in order of increasing value.
To prove this, first recall the set $I^+$ defined
in the sampling procedure; note that $|I^+|$ is
equal to the number of successes in $\frac{n}{8}$
i.i.d.\ Bernoulli trials with success probability
$\frac{1}{n \, \ptz(n)}$. Hence $\Ex{|I^+|} = \frac{1}{8 \ptz(n)}$
and, by the Chernoff Bound,
\[
\Pr{|I^+| < \tfrac{1}{16 \ptz(n)}} < 
\exp \left( - \frac{1}{128 \ptz(n)} \right) <
128 \ptz(n).
\]
Conditional on the event that $|I^+| \geq \frac{1}{16 \ptz(n)}$,
the conclusion of the proof is similar to the 
conclusion of the proof of \lemref{lem:npdacf}.
Consider any algorithm and let $t$ denote the time when
it makes its selection when the items are presented in
the order $\natperm'$. Also, let $s$ denote the time
when item $n$ is presented in the order $\natperm$.
The first $s$ items in $\natperm$ and 
$\natperm'$ have exactly the same relative
ordering by value since, by construction, 
$s \in I^+$ and hence the element that arrives 
at time $s$ in $\natperm'$ has the maximum value
observed so far. Hence, the algorithm
makes a correct selection only when $t=s$. 
However, if $t \not\in I^+$ then this event
does not happen, and if $t \in I^+$ the 
event $t=s$ happens only if $s$ is the random index
$i$ selected in Step~\ref{step:acs4}, an event
whose probability is $1/|I^+|$, which is at most
$16 \ptz(n)$ since we are conditioning on 
$|I^+| \geq \frac{1}{16 \ptz(n)}$.
Combining our bounds for the cases 
$|I^+| < \frac{1}{16 \ptz(n)}$
and
$|I^+| \geq \frac{1}{16 \ptz(n)}$,
we find that for any algorithm $\algo$,
\begin{align*}
\pcs{\algo,\natperm \revperm} & \leq
\Pr{|I^+| < \tfrac{1}{16 \ptz(n)}} + 
\Pr{|I^+| \geq \tfrac{1}{16 \ptz(n)}} \cdot
\Pr{\text{correct selection} \growingmid |I^+| \geq \tfrac{1}{16 \ptz(n)}} 
 \\ & \leq
\Pr{|I^+| < \tfrac{1}{16 \ptz(n)}} + 
\Pr{\text{correct selection} \growingmid |I^+| \geq \tfrac{1}{16 \ptz(n)}} 
 \\ & < 128 \ptz(n) + 16 \ptz(n),
\end{align*}
as desired.
\end{proof}

Finally, we present the construction of the permutation 
distribution $\npde$. A crucial ingredient is 
a coding-theoretic construction that may be of independent 
interest.

\begin{definition} \label{def:half-uniq-decoding}
We say that a function $\ecc : \{0,1\}^k \to \{0,1\}^n$ 
has half-unique-decoding radius $r$ if at least half
of the inputs $x \in \{0,1\}^k$ satisfy the property
that for all $x' \neq x$, the Hamming distance from
$\ecc(x)$ to $\ecc(x')$ is greater than $2r$.
\end{definition}

Codes with half-unique-decoding radius $r$ are
useful because they allow a receiver to decode
messages with probability at least $\frac12$,
in a model with random messages and adversarial
noise. The following easy lemma substantiates
this interpretation. Here and subsequently, we 
use $\| y - y' \|$ to denote the Hamming distance
between strings $y,y'$.

\begin{lemma} \label{lem:half-uniq-decoding}
Suppose that:
\begin{itemize}
\item a uniformly random string $x \in \{0,1\}^k$
is encoded using a function $\ecc$ whose 
half-unique-decoding radius is $r$, 
\item an adversary is allowed to corrupt any
$r$ bits of the resulting codeword, and
\item a decoding algorithm receives the corrupted
string $\hat{y}$, finds the nearest codeword
(breaking ties arbitrarily),
and applies the function $\ecc^{-1}$ to 
produce an estimate $\hat{x}$ of the original
message.
\end{itemize}
The $\Pr{\hat{x}=x} \geq \frac12$ regardless
of the adversary's policy for corrupting the
transmitted codeword.
\end{lemma}
\begin{proof}
Let $y = \ecc(x)$. The constraint on the adversary
implies that $\| y - \hat{y} \| \leq r$.
The definition of half-unique-decoding 
radius implies that, with probability at least $\frac12$
over the random choice of $x$, the nearest codeword to 
$y$ is at Hamming distance greater than $2r$. By the
triangle inequality, this event implies that $y$ is the
unique nearest codeword to $\hat{y}$, in which case the
decoder succeeds.
\end{proof}

The particular coding construction that our proof 
requires is a code with the property that, roughly
speaking, all of its low-dimensional projections
have large half-unique-decoding radius. The following
definition and lemma make this notion precise.

\begin{definition} \label{def:s-proj}
For $S \subseteq [n]$, let $\proj_S : \{0,1\}^n \to \{0,1\}^{|S|}$
denote the function that projects a vector
onto the coordinates indexed by $S$. In other words,
letting $(i_1,i_2,\ldots,i_s)$ denote a sequence containing
each element of $S$ once, we define
$\proj_S(y) = (y_{i_1},y_{i_2},\ldots,y_{i_s})$.
For any function $\ecc : \{0,1\}^k \to \{0,1\}^n$, we 
introduce the notation $\ecc_S$ to denote the 
composition $\proj_S \cdot \ecc : X \to \{0,1\}^{|S|}$.
\end{definition}

\begin{lemma} \label{lem:coding}
For all sufficiently large $m$, if $2m \leq n < \frac{m}{2} \cdot 2^{2^{m-4}}$, 
there exists a function $\ecc : \{0,1\}^m \to
\{0,1\}^n$ such that for every set $S \subseteq
[n]$ of cardinality $2m$, the function 
$\ecc_S$ has half-unique-decoding radius 
 $\frac{m}{10}$.
\end{lemma}
\begin{proof}
We prove existence of $\ecc$ using the probabilistic
method, by showing that a uniformly random function
$\ecc : \{0,1\}^m \to \{0,1\}^n$ has the property with
positive probability. To do so, we need to estimate the
probability, for a given set $S$, that 
$\ecc_S$ fails to have half-unique-decoding
radius $\frac{m}{10}$. 

Define a graph $G_S$ with vertex set $\{0,1\}^m$ by 
drawing an edge between every two vertices $x,x'$
such that $\| \ecc_S(x) - \ecc_S(x') \| \leq \frac{m}{5}$.
The event
that $\ecc_S$ has half-unique-decoding
radius $\frac{m}{10}$ corresponds precisely to
the event that $G_S$ has at least $2^{m-1}$ 
isolated vertices. When this event does not 
happen, the number of connected components in
$G_S$ is at most $2^m - 2^{m-2}$, so a spanning
forest of $G_S$ has at least $2^{m-2}$ edges.

Our plan is to bound---for every set $S \subseteq [n]$ of size
$2m$ and every forest $F$ with $2^{m-2}$ edges---the 
probability that $G_S$ contains 
all the edges of $F$. Summing over $S$ and $F$ we will
find the sum is less than 1, which implies, by the
union bound, that with positive probability over the
random choice of $\ecc$ no such pair $(S,F)$ exists.
By the arguments in the preceding paragraph, it follows
that when no such pair $(S,F)$ exists the half-unique-decoding
radius of $\ecc_S$ is $\frac{m}{10}$ for every $S$ of size
$2m$, yielding the lemma.

To begin, let us fix $x,x' \in \{0,1\}^m$
and $S \subseteq [n]$ with $|S|=2m$, and let us estimate
the probability that
$\| \ecc_S(x) - \ecc_S(x') \| \leq \frac{m}{5}$.
The strings
$\ecc_S(x)$ and $\ecc_S(x')$ are independent
uniformly-random binary strings of length $2m$.
The number of binary strings within Hamming distance
$\frac{m}{5}$ of $\ecc_S(x)$ is bounded above by
 $2^{(1+o(1)) \cdot H(1/10) \cdot 2m}$, 
where $H(p)$ denotes the binary entropy function
$-p \log_2(p) - (1-p) \log_2(1-p)$.
Using the fact that $2 H(\frac{1}{10}) < 0.95$
we can conclude that for large enough $m$, fewer
than $2^{(0.95)m}$ binary strings belong to the
Hamming ball of radius $\frac{m}{5}$ around
$\ecc_S(x)$. Hence the probability that 
$\ecc_S(x')$ is one of these strings is less
than $2^{-m/20}$. If $F$ is the edge set of a 
forest on vertex set $\{0,1\}^m$, then the
random variables $\ecc_S(x) - \ecc_S(x')$
are mutually independent 
as $(x,x')$ ranges over the edges of $F$.
Consequently the probability that all the
edges of $F$ are contained in $G_S$ is 
less than $\big( 2^{-m/20} \big)^{|F|}$.

Let $N = 2^m$. The number of spanning trees
of an $N$-element vertex set is $N^{N-2}$
\citep{Borchardt,Cayley}
and the number of forests with $N/4$ edges
contained in any one such tree is 
$\binom{N}{N/4}$. Thus, the number of 
pairs $(S,F)$ where $S \subseteq [n]$ 
has $2m$ elements and $F$ is the edge
set of a forest with vertex set $\{0,1\}^m$
is bounded above by $\binom{n}{2m}
\binom{N}{N/4} N^{N-2}$. Applying the
union bound, we conclude that the
probability of failure for our construction
is bounded above by
\[
\binom{n}{2m} \binom{N}{N/4} N^{N-2} \big( 2^{-m/20} \big)^{N/4}
<
\big( \tfrac{2n}{m} \big)^{2m} 2^N N^{N} \big( 2^{-m/20} \big)^{N/4},
\]
where we have used the inequalities $\binom{n}{k} \leq 
\big( \tfrac{4n}{k} \big)^{k}$ (valid for all $0 \leq k \leq n$)
and $\binom{N}{N/4} \leq 2^N$ (valid for all $N$).
The base-2 logarithm of the probability of failure is
bounded above by
\[
2m [1 + \log(n) - \log(m) ] + N \left[ 1 + \log(N) - \tfrac{m}{80} \right].
\]
(All logs are base 2.) Substituting $N = 2^{m-2}$ and
rearranging terms, we find that this expression is negative 
(i.e., the probability of failure is strictly less than 1) when 
\[
\log(n) < \log(m) - 1 + \frac{2^{m-3}}{m} \left[ \frac{79m}{80} - 1 \right] 
< \log(m) - 1 + 2^{m-4},
\]
provided $m > 2$. This inequality is satisfied when 
$n < \frac{m}{2} \cdot 2^{2^{m-4}}$, which completes
the proof.
\end{proof}

We now continue with the construction of the permutation
distribution $\npde$. Let 
$m = \left\lfloor \frac12 n \, \ptz(n) \right\rfloor$.
By \lemref{lem:coding} there exists a 
function $\ecc : \{0,1\}^m \to \{0,1\}^{n/8}$
such that for all $S \subseteq [n]$ with $|S| = 2m$,
the half-unique-decoding radius of $\ecc_S$ is 
$\frac{m}{10}$. Let us choose one such function $\ecc$
for the remainder of the construction. Define an
{\em $\ecc$-augmented circuit} to be a circuit
constructed from the usual {\sc and, or, not} gates
along with $n/8$ additional types of gates that 
take an $m$-bit input $x$ and output one of the bits
of $\ecc(x)$.
By \lemref{lem:very-hardf} there exists a 
function $\hardf : \{0,1\}^m \to \big[ \frac{n}{2} \big]$
such that no $\ecc$-augmented circuit of size $s(n) < 2^{m}/(4 m^2 n)$ 
computes the value of $\hardf$ correctly on more than $\frac{4}{n}$ 
fraction of inputs. Let us choose one such function and
denote it by $\hardf$ for the remainder of the construction.
(To justify the application of \lemref{lem:coding},
note that our assumption that $n \, \ptz(n) / \log(n) \to \infty$
implies $\frac{n}{8} < 2^{2^{m-4}}$ for all sufficiently large $n$.)
Armed with the functions $\hardf$ and $\ecc$ we are ready
to present the construction of $\npde$.

\begin{algorithm}[H]
\caption{Sampling procedure for $\npde$}
\label{alg:npde}
\begin{algorithmic}[1]
\State 
Sample $x \in \{0,1\}^m$ uniformly at random.
\State
Let $y = \ecc(x) \in \{0,1\}^{n/8}$.
\For{$i = 1,\ldots,\frac{n}{8}$}
\If{$y_i=1$}
\State Swap the items in positions $\frac{3n}{8}+i$
and $\frac{n}{4}+i$.
\Else
\State Leave the permutation unchanged.
\EndIf
\EndFor
\State Swap the items in positions $n$ and $\frac{n}{2} + \hardf(x)$.
\end{algorithmic}
\end{algorithm}

The corresponding algorithm $\algoe$ works as follows.
\begin{algorithm}[H]
\caption{Algorithm $\algoe$}
\begin{algorithmic}[1]
\State 
Observe the first $\frac{n}{2}$ elements of the input sequence.
\State \label{step:algoe2}
Let $J$ denote the set of items with arrival times in the 
interval $\big( 
\frac{n}{4} , \frac{n}{2} \big]$, i.e.\ 
$J = \natperm^{-1} \big( \frac{n}{4} , \frac{n}{2} \big].$
\State 
Let $j_1,\ldots,j_{n/4}$ denote a listing of the elements of $J$
in order of decreasing value.
\For{$\ell = 1,\ldots,2m$}
\If{$\natperm(j_\ell) \leq \frac{3n}{8}$}
\State Set $\hat{y}_\ell = 1$ and $i_\ell = \natperm(j_\ell) - \frac{n}{4}$.
\Else
\State Set $\hat{y}_\ell = 0$ and $i_\ell = \natperm(j_\ell) - \frac{3n}{8}$.
\EndIf
\EndFor
\State Set $S = (i_1,i_2,\ldots,i_{2m})$.
\State Find the $x \in \{0,1\}^m$ that minimizes $\| \ecc_S(x) - \hat{y} \|$,
breaking ties arbitrarily.
\State Select the item that arrives at time $t = \frac{n}{2} + g(x)$.
\end{algorithmic}
\end{algorithm}

\begin{lemma} \label{lem:npde}
If the adversary's ordering $\advperm$ is
a decodable ordering, then 
$\gval[\npde](\algoe,\advperm) \geq \frac12$.
On the other hand, 
for any algorithm $\algop$ whose stopping rule
can be computed by circuits of size $s(n) = 2^{n \, \ptz(n) / 4}$,
we have $\gval[\npdacf](\ast,\revperm) \leq 4/n$.
\end{lemma}
\begin{proof}
Note that a permutation $\natperm$ sampled from $\npde$
always maps the set $\big( \frac{n}{4}, \, \frac{n}{2} \big]$ to itself,
though it may permute the elements of that set. Consequently, when one
runs $\algoe$ on an input sequence ordered using $\natperm$ in the
support of $\npde$, it sets $J = \big( \frac{n}{4}, \, \frac{n}{2} \big]$.
The definition of a decodable permutation now implies that
the fraction of items in $\{j_1,\ldots,j_{2m}\}$ that belong
to $\big( \frac{3n}{8}, \, \frac{n}{2} \big]$ is at least 
$\frac{39}{40}$; let us call the remaining items in 
$\{j_1,\ldots,j_{2m}\}$ ``misplaced''. 
For each $j_\ell$ that is not misplaced, $\algoe$ correctly
deduces the corresponding value $y_\ell$ {\em unless} item
$j_{\ell} - \frac{n}{8}$ also belongs to $\{j_1,\ldots,j_{2m}\}$
(in which case it is a misplaced item). Hence each misplaced
item contributes to potentially two errors, meaning that 
at most $\frac{1}{20}$ fraction of the bits in $\hat{y}$ 
differ from the corresponding bit in $\ecc(x)$. These strings
have length $2m$, so we have shown their Hamming distance
is at most $\frac{m}{10}$. \lemref{lem:half-uniq-decoding}
now ensures that with probability at least $\frac12$,
$\algoe$ decodes the appropriate value of $x$. When this
happens, it correctly selects item $n$ from the second half
of the input sequence. Our assumption that $\advperm$ is 
decodable means that item $n$ is the item with maximum
value, which completes the proof that 
$\gval[\npde](\algoe,\advperm) \geq \frac12$.

To prove the second statement in the lemma,
we can use $\algop$ to guess the value of $\hardf(x)$ for any
input $x \in \{0,1\}^{m}$ by the following simulation
procedure. First, define a permutation $\natperm'(x)$ by 
running Algorithm~\ref{alg:npde} with random string $x$,
omitting the final step of swapping the items in positions
$n$ and $n/2 + \hardf(x)$; note that this means that $\natperm'(x)$,
unlike $\natperm(x)$, can be constructed from input $x$ by an
$\ecc$-augmented circuit of polynomial size.
Now simulate $\algo$
on the input sequence $\natperm'(x)$, observe the
time $t$ when it selects an item, and output
$t - \frac{n}{2}$. The $\ecc$-augmented
circuit complexity of this simulation
procedure is at most $\operatorname{poly}(n)$ times
the $\ecc$-augmented circuit complexity of the stopping rule implemented
by $\algo$, and the fraction of inputs $x$ on which it guesses
$\gval(x)$ correctly is precisely $\gval(\algo,\idpd)$. 
(To verify this last statement, note that $\algo$  
makes its selection at time $t = \frac{n}{2} + \hardf(x)$ when
observing input sequence $\natperm(x)$ {\em if and only if}
if also makes its selection at time $t$ when observing
input sequence $\natperm'(x)$, because the two input 
sequences are indistinguishable to comparison-based
algorithms at that time.)
Hence,
if $\gval(\algo,\idpd) > \frac{4}{n}$ then the stopping rule
of $\algo$ cannot be implemented by circuits of size $2^{m/2} = 
2^{n \, \ptz(n) / 4}$.
\end{proof}
\subsection{Proofs deferred from \secref{extension:sec}}
\subsubsection{Full proof of Theorem~\ref{kunif:thm}\secref{kunif:sec}}
We start by proving the following lemmas which turn out to be critical for  the analysis of $\alg{\items}{k}{q}$ under non-uniform permutation distributions. In fact, these lemmas capture the fact that if membership random variables of different items for the random set $\sam$ (and $\nsam$) are almost pairwise independent (rather than mutually independent), then we still preserve enough of probabilistic properties that are needed in the analysis of algorithm proposed by~\cite{Kleinberg05}.
\begin{lemma}
\label{k-unif:lemma1}
Suppose $\p$ is drawn from a permutation distribution satisfying $(p,q,\delta)$-BIP for $p\geq 2$ and $\sam\triangleq\{x\in \items:\bind(x)\leq \thresh(q)\}$ where $\thresh(q)$ is independently drawn from $\textrm{Binom}(q,1/2)$. Then for any $T\subseteq \items$ such that $\delta\leq \frac{1}{\sqrt {\lvert T \rvert} }$  we have
\begin{enumerate}
\item $\ex{\lvert T\cap \sam\rvert}\in [(1-\delta)\lvert T\rvert/2,(1+\delta)\lvert T\rvert/2]$
\item $\ex{\valfun(T\cap\sam)}\in[(1-\delta)\valfun(T)/2,(1+\delta)\valfun(T)/2]$
\item $\pr{\lvert T\cap \sam\rvert\geq\lvert T\rvert/2+\alpha}\leq \frac{\lvert T\rvert}{2\alpha^2} $
\end{enumerate}
\begin{proof}
For $x\in\items$, let $Y_x$ be a 0/1 variable indicating if $x \in S$. We have 
\begin{align*}
\Ex{Y_x} &= \sum_{i=1}{q} \Pr{\text{$x$ is in block $i$}} \Pr{\thresh_b \geq i} 
\\ & \geq (1-\delta) \frac{1}{q} \sum_{i=1}{q} \Pr{\thresh_b \geq i} \\
& = (1-\delta) \frac{1}{q} \Ex{\thresh_b} \\
& = (1 - \delta) \frac{1}{q} \frac{q}{2} \\
& = \frac{1 - \delta}{2}.
\end{align*} 
Analogously, we get $\Ex{Y_x} \leq \frac{1 + \delta}{2}$.

Claims 1 and 2 now follow from linearity of expectation, e.g., $\Ex{\lvert T\cap \sam\rvert} = \Ex{\sum_{x \in T} Y_x} \geq \frac{1-\delta}{2} \lvert T \rvert$.

To show Claim 3, we use that for $x \neq x'$, we have $\Ex{Y_x Y_{x'}} \leq \frac{1 + \delta}{4}$. This implies $\Ex{\lvert T \cap \sam\rvert^2} = \Ex{\sum_x Y_x} + \Ex{\sum_{x \neq x'} Y_x Y_{x'}} \leq \Ex{\lvert T\cap \sam\rvert} + \lvert T \rvert(\lvert T \rvert - 1) \frac{1 + \delta}{4} \leq \Ex{\lvert T\cap \sam\rvert}$.

By Markov's inequality, we get
\[
\Pr{\lvert T\cap \sam\rvert\geq\lvert T\rvert/2+\alpha} \leq \Pr{\left( \lvert T\cap \sam\rvert - \lvert T\rvert/2 \right)^2 \geq \alpha^2} \leq \frac{1}{\alpha^2} \Ex{\left( \lvert T\cap \sam\rvert - \lvert T\rvert/2 \right)^2} \enspace.
\]

Using linearity of expectation and the bounds obtained so far, we get 
\begin{align*}
\Ex{\left( \lvert T\cap \sam\rvert - \lvert T\rvert/2 \right)^2} & = \Ex{\left( \lvert T\cap \sam\rvert \right)^2} - \lvert T\rvert \Ex{\lvert T\cap \sam\rvert} + \left(\frac{\lvert T\rvert}{2} \right)^2 \\
& \leq \frac{1 + \delta}{2} \lvert T \rvert^2 - \frac{1 + \delta}{4} \lvert T \rvert - (\lvert T \rvert - 1) \Ex{\lvert T\cap \sam\rvert} \\
& \leq \frac{1 + \delta}{2} \lvert T \rvert^2 - \frac{1 + \delta}{4} \lvert T \rvert - (\lvert T \rvert - 1) \frac{1 - \delta}{2} \lvert T \rvert \\
& \leq \delta \lvert T \rvert^2 + \frac{1-3\delta}{4}\lvert T\rvert\leq \frac{\lvert T\rvert}{2}  \enspace.
\end{align*}
where the last inequality is true because of $\delta\leq \frac{1}{\sqrt {\lvert T \rvert} }$.
\end{proof}
\end{lemma}
\begin{lemma}
\label{k-unif:lemma2}
Suppose $\p$ is drawn from a permutation distribution satisfying $(p,q,\delta)$-BIP for some $p\geq 2$ and $\sam$ is as defined in Lemma~\ref{k-unif:lemma1}. Let $Y_1$ be the (possibly negative) random variable such that $(k/2)^{\textrm{th}}$ item in the sorted-by-value list of items in $\sam$ is the $(k+Y_1)^{\textrm{th}}$ in the sorted-by-value list of items in $\items$. Then $\ex{\lvert Y_1 \rvert }=O(\sqrt k)$.
\end{lemma}
\begin{proof}
We have $\ex{\lvert Y_1\rvert }=\sum_{i=1}^{\infty}\pr{\lvert Y_1\rvert\geq i}=\sum_{i=1}^{\infty} \pr{Y_1\geq i}+\sum_{i=1}^{\infty}\pr{Y_1\leq -i} $. Now we bound each of the terms separately. For a fixed $i$, look at the event $Y_1\leq -i$. This event is equivalent to the event that the number of items in $\sam$ among $k-i$ highest-valued items is at least $k/2$. Let us define $r\triangleq k-i$. Furthermore, let $T_r$ be the set of the $r$-highest valued items. Using Lemma~\ref{k-unif:lemma1} we have:
\begin{equation}
\pr{ Y_1 \leq -i}=\pr{\lvert T_r\cap\sam\rvert \geq r/2+i/2}\leq \frac{2(k-i)}{i^2}
\end{equation}
So we have 
\begin{align}
\sum_{i=1}^{\infty}\pr{Y_1 \leq -i}&\leq \sum_{i=1}^{\ceil{\sqrt k}}1+\sum_{i=\ceil {\sqrt k} +1}^{k}\frac{2(k-i)}{i^2}\leq 1+\sqrt k+2k\sum_{i=\ceil{\sqrt k}+1}^{\infty}\frac{1}{i^2}\nonumber\\
&\leq 1+\sqrt k+2k\int_{\sqrt k}^{\infty}\frac{1}{x^2}dx=3\sqrt k+1=O(\sqrt k)
 \end{align}

Now, let's consider the event $Y_1\geq i$. This event implies that number of items of $\nsam$ among the $k$ highest valued items is at least $k/2+i$. Again, using Lemma~\ref{k-unif:lemma1} we have:
\begin{equation}
\pr{ Y_1 \geq i}\leq\pr{\lvert T_{k}\cap\nsam\rvert \geq k/2+i}\leq \frac{k}{2i^2}
\end{equation}
and hence we have 
\begin{align}
\sum_{i=1}^{\infty}\pr{ Y_1 \geq i}&\leq \sum_{i=1}^{\ceil{\sqrt k}}1+\sum_{i=\ceil {\sqrt k} +1}^{\infty}\frac{k}{2i^2}\leq 1+\sqrt k+\frac{k}{2}\sum_{i=\ceil{\sqrt k}+1}^{\infty}\frac{1}{i^2}\nonumber\\
&\leq 1+\sqrt k+\frac{k}{2}\int_{\sqrt k}^{\infty}\frac{1}{x^2}dx=\frac{3}{2}\sqrt k+1=O(\sqrt k)
 \end{align}
 which completes the proof.
 \end{proof}

 Now we start proving the theorem. Basically, we prove for any fixed $k$ there exists a function $\epsilon(k,q,\delta)$, non-increasing w.r.t. $q$, such that $\alg{\items}{k}{q}$ is $\left(1-O(\frac{1}{k^{\frac{1}{3}}})-\epsilon(k,q,\delta)\right)$-competitive, and $\epsilon$ goes to $0$ as $q\rightarrow \infty$ and $\delta \rightarrow 0$ for a fixed $k$.
 First without loss of generality we modify values so that if the value is among  $k$ highest it remains the same, otherwise it is set to $0$. This doesn't change sum of the values of the $k$ highest items, and just weakly decreases the values of items picked by any algorithm. Now, run the algorithm with modified values. Let $\algset$ be the set of items picked by $\alg{\items}{k}{q}$ and $\opt$ be the subset of $k$ highest value items under $\sig$. Define $\sam\triangleq\{x\in \items:\pi(x)\leq \frac{\tau(q)}{q}n\}$ to be the set sampled before threshold and $\nsam\triangleq \items \backslash \sam$ be its complement. Suppose $v_0$ is the value of the ${\frac{k}{2}}^{\textrm{th}}$  highest valued item in $\sam$ (if $\lvert \sam\rvert<\frac{k}{2}$, set $v_0=0$).  Moreover, define the value function $\valfun(.)$ to be the sum of values of the input set of items under $\sig$. 
 
 Fixing $\sig$, we prove the claim by induction over $k$.  The case $k=1$ is exactly the case of a single secretary, which is analyzed in \secref{section:secretaryanalysis}. For general $k$, we first run $\alg{\items\cap\sam}{k/2}{\tau(q)}$ to give us $\algset\cap\sam$. Note that the ordering of arrivals of items in $\sam$ satisfies $(p,\tau(q),\delta)$-BIP. So, by induction hypothesis and conditioned on set $\sam$ we have
 \begin{align}
 \ex{\valalgsam|\sam}\geq \ex{\valoptsamt\left(1-O(\frac{1}{k^{\frac{1}{3}}})-\epsilon(k/2,\tau,\delta)\right)|\sam}
 \end{align}
We can lower-bound the right-hand side further as follows.
\begin{align}
&\ex{\valoptsamt\left(1-O(\frac{1}{k^{\frac{1}{3}}})-\epsilon(k/2,\tau,\delta)\right)|\sam}\geq\nonumber\\
& \ex{\valoptsamt|\sam} -\valfun(\opt)\left(O(\frac{1}{k^{\frac{1}{3}}})+\ex{\epsilon(k/2,\tau,\delta)|\sam}\right), 
\end{align}
and by taking expectation with respect to $\sam$ we have
\begin{align}
&\ex{\valalgsam} \geq \ex{\valoptsamt} -\valfun(\opt)\left(O(\frac{1}{k^{\frac{1}{3}}})+\ex{\epsilon(k/2,\tau,\delta)}\right)\label{eqaution17}
\end{align}
 Now suppose $\ind{.}$ is the indicator function. One can easily decompose $\epsilon(k/2,\tau,\delta)$ as follows.
 \begin{align}
 \label{equation18}
 \epsilon(k/2,\tau,\delta)&=\epsilon(k/2,\tau,\delta)\ind{\tau\geq q/4}+\epsilon(k/2,\tau,\delta)\ind{\tau<q/4}\leq \epsilon(k/2,q/4,\delta)+ \ind{\tau<q/4}
 \end{align}
 where the last inequality is true because $\epsilon(k/2,\tau,\delta)$ is non-increasing w.r.t. $\tau$. Now by taking expectations from both hand sides of  (\ref{equation18}), we have:
 \begin{align}
 \ex{\epsilon(k/2,\tau,\delta)}&\leq \epsilon(k/2,q/4,\delta)+ \pr{\tau<q/4}\nonumber\\
 &=\epsilon(k/2,q/4,\delta)+\pr{\tau<q/2(1-1/2)}\leq \epsilon(k/2,q/4,\delta)+e^{-\frac{q}{16}}\label{equation19}
 \end{align}
 where in the last inequality we used  Chernoff bound, as $\tau$ is drawn from $\textrm{Binom}(q,1/2)$.
 Now, fix $\eps'$. For a given $\eps'$ we have
 \begin{align}
 &\ex{\valfun([\opt\cap\sam]_{k/2})}=\ex{\valfun(\opt\cap\sam)\ind{\lvert\opt\cap\sam\rvert<k/2}}
 +\ex{\valfun([\opt\cap\sam]_{k/2})\ind{\lvert\opt\cap\sam\rvert\geq k/2}}\nonumber\\
 &\geq \ex{\valfun(\opt\cap\sam)\ind{\lvert\opt\cap\sam\rvert<k/2}}+\ex{\frac{k/2}{\lvert\opt\cap\sam\rvert}\valfun(\opt\cap\sam)\ind{\lvert\opt\cap\sam\rvert\geq k/2}}\nonumber\\
&\geq \ex{\valfun(\opt\cap\sam)\ind{\lvert\opt\cap\sam\rvert<k/2}}+\frac{1}{1+\eps'}\ex{\valfun(\opt\cap\sam)\ind{\frac{k}{2}(1+\eps')\geq\lvert\opt\cap\sam\rvert\geq k/2}}\label{k-unif:eq1}
\end{align}
Also, we have:
\begin{align}
&\frac{1}{1+\eps'}\ex{\valfun(\opt\cap\sam)\ind{\frac{k}{2}(1+\eps')\geq\lvert\opt\cap\sam\rvert\geq k/2}}\nonumber\\
&=\frac{1}{1+\eps'}\ex{\valfun(\opt\cap\sam)\ind{\lvert\opt\cap\sam\rvert\geq k/2}}-\frac{1}{1+\eps'}\ex{\valfun(\opt\cap\sam)\ind{\lvert\opt\cap\sam\rvert\geq \frac{k}{2}(1+\eps')}}\nonumber\\
&\geq \frac{1}{1+\eps'}\ex{\valfun(\opt\cap\sam)\ind{\lvert\opt\cap\sam\rvert\geq k/2}}-\valfun(\opt)\pr{\lvert\opt\cap\sam\rvert\geq \frac{k}{2}(1+\eps')}\nonumber\\
&\overset{(1)}{\geq}  \frac{1}{1+\eps'}\ex{\valfun(\opt\cap\sam)\ind{\lvert\opt\cap\sam\rvert\geq k/2}}-\frac{1}{k{\eps'}^2}\valfun(\opt)
\label{k-unif:eq2}
 \end{align}
in which (1) is true because of Lemma~\ref{k-unif:lemma1}. Combining (\ref{k-unif:eq1}) with (\ref{k-unif:eq2}) we have:
\begin{equation}
\ex{\valfun([\opt\cap\sam]_{k/2})}\geq \frac{1}{1+\eps'}\ex{\valoptsam}-\frac{1}{k{\eps'}^2}\valfun(\opt)\geq \ex{\valoptsam}-(\eps'+\frac{1}{k{\eps'}^2})\valfun(\opt)\label{equation22}.
\end{equation}
Finally, by combining (\ref{eqaution17}), (\ref{equation19}) and (\ref{equation22}) and setting $\eps'=\frac{1}{k^{\frac{1}{3}}}$, we have
\begin{equation}
\label{equation23}
\ex{\valalgsam} \geq \ex{\valoptsam}-\valfun(\opt)\left(O(\frac{1}{k^{\frac{1}{3}}})+e^{-\frac{q}{16}}+\epsilon(k/2,q/4,\delta)\right)
\end{equation}

Next, we try to lower-bound $\ex{\valalgnsam}$ by $\ex{\valoptnsam}$. Lets define random variable  $\algnumsam\triangleq\lvert\algset\cap\nsam\rvert$ to be number of items algorithm picked from $\nsam$. We have $\ex{\valalgnsam}=\sum_{x=0}^{k/2}\ex{\valalgnsam\ind{Q=x}}$.  Now we look at each term $\ex{\valalgnsam\ind{Q=x}}$, and we try to lower-bound it with $\ex{\valoptnsam\ind{Q=x}}$ for different values of $x$. Consider two cases:\\

 \emph{Case 1, when $x<\frac{k}{2}:$} In this case $v_0>0$ and all items in $\sam$ with value more than $v_0$ are in $\opt$. We know the number of items in $\nsam$ that have value at least $v_0$ is $x$. If we look at items in $\opt\cap\nsam$, all items in $\algset\cap\nsam$ are also in $\opt\cap\nsam$ and in addition we have at most $k-(k/2+x)=k/2-x$ items in $\opt\cap\nsam$, all of which have value at most $v_0$. Hence, as the value of any item in $\algset\cap\nsam$ is at least $v_0$, the followings hold deterministically :
 \begin{align}
 &\valoptnsam-\valalgnsam=\valfun(\{x\in\opt\cap\nsam: \val{x}\leq v_0\})\nonumber\\
 &\leq \frac{\lvert \opt\cap\nsam\rvert-\lvert \mathcal{A}\cap\nsam\rvert}{\lvert  \mathcal{A}\cap\nsam\rvert}\valalgnsam\leq (\frac{k}{2x}-1)\valalgnsam
 \end{align}
 which implies $\valalgnsam\geq \frac{2x}{k}\valoptnsam$ when $x<k/2$. So for $x<k/2$,
 \begin{equation}
  \label{k-unif:eq6}
 \ex{\valalgnsam \ind{Q=x}}\geq \ex{\frac{2Q}{k}\valoptsam \ind{Q=x}}
 \end{equation}
 \emph{Case 2, when $x=\frac{k}{2}$}:  In this case either $v_0>0$, which implies at least there are $k/2$ items in $\opt\cap\sam$.  As algorithm also picks $k/2$ items  and so $\algset\cap\nsam=\opt\cap\nsam$ for which we are done. Otherwise, suppose $v_0=0$. We know the permutation distribution generating $\pi$ satisfies the $(p,q,\delta)$-BIP some $p\geq k$, and hence it satisfies $(k,q,\delta)$-BIP. So, based on Theorem~\ref{theorem-BIP} it also satisfies $(k,\delta+\frac{k^2}{q})$-UIOP. Roughly speaking, if you look at any subset of elements with cardinality at most $k$, their induced ordering is almost uniformly distributed (within an error of $\delta+\frac{k^2}{q}$). We know in this case algorithm picks $\frac{k}{2}$ items (all of them in $\opt\cap \nsam$), and in fact it picks the first $\frac{k}{2}$ elements of $\opt\cap \nsam$ in the ordering of elements in $\opt\cap \nsam$ induced by the permutation $\pi$. Suppose $X\triangleq\lvert\opt\cap\nsam\rvert-k/2$. Then 
 \begin{align}
 \ex{\valalgnsam\ind{Q=\frac{k}{2}}}&\geq \ex{\valfun(\textrm{first $k/2$ elements of $\opt\cap\nsam$ in the ordering $\pi$})\ind{Q=\frac{k}{2}}}\nonumber\\
 &= \ex{\valoptnsam\ind{Q=\frac{k}{2}}}\nonumber\\
 &~~-\ex{\valfun(\textrm{last $X$ elements of $\opt\cap\nsam$ in the ordering  $\pi$})\ind{Q=\frac{k}{2}}}\nonumber\\
 &\geq  \ex{\valoptnsam\ind{Q=\frac{k}{2}}}\nonumber\\
 &~~-\ex{\valfun(\textrm{last $X$ elements of $\opt\cap\nsam$ in the ordering  $\pi$})}
\label{equation26}
 \end{align}
 
 For a fixed set $\sam$, we have $(k,\delta+\frac{k^2}{q})$-UIOP for elements in $\opt\cap\nsam$ (this is an order oblivious fact), and hence the induced ordering of the elements in $\opt\cap\nsam$ is almost uniform. So, we have
 \begin{align}
 &\ex{\valfun(\textrm{last $X$ elements of $\opt\cap\nsam$ in the ordering  $\pi$})|\sam}\leq (1+\delta+\frac{k^2}{q})\ex{\valoptnsam \frac{X}{\lvert\opt\cap\nsam \rvert}|\sam}\nonumber\\
 &\leq (1+\delta+\frac{k^2}{q})\textrm{val}(\opt)\ex{\frac{\lvert\opt\cap\nsam\rvert-k/2}{\lvert \opt\cap\nsam\rvert}|\sam}
 \label{equation28}
 \end{align} 
 
 Now by taking expectation w.r.t. $\sam$ and combining it with (\ref{equation26}) we have 
 \begin{equation}
 \ex{\valalgnsam\ind{Q=\frac{k}{2}}}\geq  \ex{\valoptnsam\ind{Q=\frac{k}{2}}}-(1+\delta+\frac{k^2}{q})\textrm{val}(\opt)\ex{\frac{\lvert\opt\cap\nsam\rvert-k/2}{\lvert \opt\cap\nsam\rvert}}
 \end{equation}
 Moreover, one can use Lemma~\ref{k-unif:lemma1} to find an upper-bound on the error term in (\ref{equation28}). Fix any $\eps'$, Now we have 
 \begin{equation}
 \frac{\lvert\opt\cap\nsam\rvert-k/2}{\lvert \opt\cap\nsam\rvert}=\frac{\lvert\opt\cap\nsam\rvert-k/2}{\lvert \opt\cap\nsam\rvert}\ind{\lvert\opt\cap\nsam \rvert< k/2+\eps'}+\frac{\lvert\opt\cap\nsam\rvert-k/2}{\lvert \opt\cap\nsam\rvert}\ind{\lvert\opt\cap\nsam \rvert\geq k/2+\eps'}
 \label{equation29}
 \end{equation}
 By taking expectation from both sides of (\ref{equation29}), setting $\eps'=k^{\frac{1}{3}}$ and using Lemma~\ref{k-unif:lemma1} we have 
 \begin{align}
 \ex{\frac{\lvert\opt\cap\nsam\rvert-k/2}{\lvert \opt\cap\nsam\rvert}}&\leq \frac{\eps'}{k/2+\eps'}+\pr{\lvert\opt\cap\nsam \rvert\geq k/2+\eps'}\leq \frac{\eps'}{k/2}+\frac{k}{2\eps'}\leq \frac{3}{k^{\frac{1}{3}}}.
 \label{equation30}
 \end{align}
By combining (\ref{equation30}) and (\ref{equation28}) we have (note that $\delta \leq 1$) 
\begin{align}
\ex{\valalgnsam \ind{ Q=k/2}}\geq \ex{\frac{2Q}{k}\valoptnsam \ind{ Q=k/2}}-\valfun(\opt)(\frac{6}{k^{\frac{1}{3}}}+\frac{3k^{\frac{5}{3}}}{q})
\end{align}
 As we desired. 
 
 Now, by combining the above cases with each other we have
 \begin{align}
 \label{equation32}
 \ex{\valalgnsam}&\geq  \ex{\frac{2Q}{k}\valoptnsam}-\valfun(\opt)(\frac{6}{k^{\frac{1}{3}}}+\frac{3k^{\frac{5}{3}}}{q})\nonumber\\
&\geq  \ex{\valoptnsam}-\valfun(\opt)\left(\frac{6}{k^\frac{1}{3}}+\frac{3k^{\frac{5}{3}}}{q}+\ex{\frac{k/2-Q}{k/2}}\right)
 \end{align}
 Finally, by combining equations (\ref{equation23}) and (\ref{equation32}) we have 
 \begin{equation}
  \ex{\valfun(\algset)}\geq \valfun(\opt)\left(1-\ex{\frac{k/2-Q}{k/2}}-O(\frac{1}{k^{\frac{1}{3}}})-\frac{3k^{\frac{5}{3}}}{q}-e^{\frac{-q}{16}}+\epsilon(k/2,q/4,\delta)\right)
 \end{equation}

 As it can be seen, the question of finding the competitive ratio of $1-o(1)$ boils down to upper-bounding $\ex{k/2-Q}$. To do so, we define random variable $Y_1$ such that $(k/2)^{\textrm{th}}$ item in sorted-by-value list of items in $\sam$ will be the $(k+Y_1)^{\textrm{th}}$ item in $\items$. Now we claim that $Q\geq k/2-\lvert Y_1\rvert$. The proof is easy. If $v_0=0$ then algorithm picks $k/2$ items from $\nsam$ and we are done. Otherwise, there are $k+Y_1-k/2=k/2+Y_1$ items in $\nsam$ such that their values is at least $v_0$. By a simple case analysis, if $Y_1\leq 0$ then algorithm picks all of those,  and hence $Q\geq k/2+ Y_1=k/2-\lvert Y_1\rvert$. If $Y_1\geq 0$ then algorithm picks $k/2$ items which is $\geq k/2-\lvert Y_1\rvert$ we are again down. So $\ex{k/2-Q}\leq \ex{\lvert Y_1\rvert}$. Lemma~\ref{k-unif:lemma2} shows that $\ex{\lvert Y_1\rvert}= O(\sqrt k)$, and hence $\ex{\frac{k/2-Q}{k/2}}\leq O(1/\sqrt k)$. Hence, we have 
\begin{align}
\ex{\valfun(\algset)}&\geq \valfun(\opt)\left(1-O(\frac{1}{k^{\frac{1}{3}}})-\frac{3k^{\frac{5}{3}}}{q}-e^{\frac{-q}{16}}-\epsilon(k/2,q/4,\delta)\right)\nonumber\\
 &\geq  \valfun(\opt)\left(1-O(\frac{1}{k^{\frac{1}{3}}})-\epsilon(k,q,\delta)\right)
 \end{align}
  which completes the proof, as $\epsilon$ can be arbitrary small for large enough $q$ and small enough $\delta$ .

\subsubsection{Full proof of Theorem~\ref{kp:thm}}
\label{sec:korula-pal}
We define a randomized construction that defines an input and a probability distribution simultaneously. By the probabilistic method this implies the statement.

Our bipartite graph has $n$ vertices on the online and the offline side each. For each pair $(j, i)$, we add the connecting edge with probability $\frac{1}{2} - 8 \sqrt{\frac{\ln n}{n}}$ independently. In case $j$ and $i$ are connected, the edge weight is set to $w(j, i) = 1 - \epsilon(j + i)$ for $\epsilon = \frac{1}{n^2}$. This way, the expected weight of the optimal solution is $\Omega(n)$.

To define the distribution over permutations $\pi_i\colon \items \to [n]$, we draw for the first $\xi = \frac{2 (k+1)!}{\delta^2} \log n$ offline vertices $i \in R$ one permutation uniformly at random from the set of all permutations in which the neighbors of node $i$ come last. Afterwards, we draw one of these permutations $\pi_1, \ldots, \pi_\xi$ at random. We claim that this way, the probability distribution fulfills the $(k, \delta)$-uniform-induced-ordering property.

Fix $k$ distinct items $x_1, \ldots, x_k \in \items$. Note that we can ignore the fact that in any permutation neighbors come last as all $x_1, \ldots, x_k$ have the same probability of corresponding to a neighbor. Therefore, we can steadily follow the argument from Theorem~\ref{thm:rlec}. Let $Y_i = 1$ if $\pi_i(x_1) < \pi_i(x_2) < \ldots < \pi_i(x_k)$. As $\pi_i$ is drawn uniformly from the set of all permutations, we have $\Pr{Y_i = 1} = \frac{1}{k!}$. That is, we have $\Ex{\sum_{i = 1}^\xi Y_i} = \frac{\xi}{k!}$. As the random variables $Y_i$ are independent, we can apply a Chernoff bound. This gives us
\[
\Pr{\sum_{i = 1}^\xi Y_i \leq (1 - \delta) \frac{\xi}{k!}} \leq \exp\left( - \frac{\delta^2}{2} \frac{\xi}{k!} \right) = n^{k+1} \enspace.
\]
Note that if $\sum_{i = 1}^\xi Y_i \leq (1 - \delta) \frac{\xi}{k!}$ then the respective sequence $x_1, \ldots, x_k \in \items$ has probability at least $(1 - \delta) \frac{1}{k!}$ when drawing one permutation from $\pi_1, \ldots, \pi_\xi$.

There are fewer than $n^k$ possible sequences. Therefore, applying a union bound, with probability at least $1 - \frac{1}{n}$ the bound is fulfilled for all sequences simultaneously and therefore $S$ fulfills the stated condition.

It now remains to show that the Korula-P\'{a}l algorithm has a poor performance on this type of instance. The algorithm draws a transition point $\tau \sim \mathrm{Binom}(n, \frac{1}{2})$, before which it only observes the input and after which it starts a greedy allocation based on the vertices seen until round $\tau$ and the current vertex. It is important to remark that for the tentative allocation the other vertices seen between round $\tau$ and the current round are ignored. Only after a tentative edge has been selected, their allocation is taken into consideration in order to check whether the matching would still be feasible.

Let $\pi_i$ be the chosen permutation. That is, the neighbors of $i$ come last. Let $i$ have $n-A$ neighbors. We now claim that with high probability no neighbor of $i$ comes before $\tau$, i.e., $A < \tau$. Furthermore, after $\tau$ essentially only neighbors of $i$ arrive. This has the consequence that almost all vertices are tentatively matched to $i$. However, only the first such edge is feasible.

Using Chernoff bounds, we get
\[
\Pr{\tau < \frac{n}{2} - 2 \sqrt{n \ln n}} = \Pr{n - \tau > \frac{n}{2}\left(1 + 4 \sqrt{\frac{\ln n}{n}}\right)} \leq \exp\left( - \frac{1}{3} 16 \frac{\ln n}{n} \frac{n}{2} \right) < \frac{1}{n} \enspace.
\]
\[
\Pr{A > \frac{n}{2} - 2 \sqrt{n \ln n}} \leq \Pr{R > \left(1 + 4 \sqrt{\frac{\ln n}{n}}\right)\left(\frac{n}{2} - 8 \sqrt{n \ln n}\right)} < \frac{1}{n} \enspace.
\]
In case $A < \tau$, the value of the solution is upper-bounded by $j$ because every node is tentatively matched to a vertex of index at most $i$. As $i \leq \tau$, this gives a value bounded by $\xi = \frac{2(k+1)!}{\delta^2} \ln n$.

\end{document}